\documentclass[11pt]{article}
 
\usepackage[utf8]{inputenc}
\usepackage{authblk}

\usepackage{fullpage}

\usepackage{tabularx}

\usepackage{amsthm}
\usepackage{amsmath}
\usepackage{amssymb}
\usepackage{amsfonts}
\usepackage{mathtools}
\usepackage{enumitem}
\usepackage[pagebackref]{hyperref}
\usepackage[sort,numbers]{natbib}

\usepackage{multicol}

\usepackage{xspace}

\usepackage{tikz}
\usetikzlibrary{arrows,decorations.pathmorphing,decorations.pathreplacing,backgrounds,positioning,fit,matrix}
\usetikzlibrary{shapes,calc,patterns,arrows.meta}
\tikzset{
	vert/.style={circle,inner sep=1.5,fill=white,draw,minimum size=.3cm},
	vert2/.style={circle,inner sep=1.5,fill=gray,draw,minimum size=.4cm},
    tvert/.style={rectangle,fill=black,minimum height=.4cm,inner sep=.5},
	edge/.style={color=black, thick},
	diredge/.style={->,>={Stealth[width=8pt,length=8pt]},color=black, thick},
	timelabel/.style={fill=white,font=\footnotesize, text centered}
}

\newcommand{\tikzpreamble}{
	\tikzstyle{xnode}=[circle,scale=\nsc,draw];
	\tikzstyle{tnode}=[diamond,scale=1.2*\nsc,fill];
	
	\tikzstyle{vnode}=[rectangle,draw,inner sep=0,minimum size=15.2*\nsc,fill];
	\tikzstyle{xnodeSmall}=[circle,scale=\nsc*0.65,draw];
	\tikzstyle{satE} = [ultra thick];
	\tikzstyle{unsatE} = [thick,dotted];
	\definecolor{lblue}{RGB}{151, 197, 232}
	\definecolor{lred}{RGB}{245, 193, 193}
}

\usepackage[capitalize]{cleveref}

\usepackage{todonotes}

\newtheorem{theorem}{Theorem}
\newtheorem{lemma}[theorem]{Lemma}

\newtheorem{corollary}[theorem]{Corollary}

\theoremstyle{definition}

\newtheorem{construction}{Construction}

\crefname{claim}{Claim}{Claims}
\crefname{figure}{Figure}{Figures}

\newcommand{\yes}{yes}
\newcommand{\no}{no}

\newcommand{\RD}{$(\Rightarrow)\:$}
\newcommand{\LD}{$(\Leftarrow)\:$}
\newcommand{\eq}{=}

\newcommand{\N}{\mathbb{N}}

\newcommand{\ts}{\tilde{s}}
\newcommand{\ttt}{\tilde{z}}
\newcommand{\tW}{\tilde{W}}
\newcommand{\tw}{\tilde{w}}
\newcommand{\tx}{\tilde{x}}
\newcommand{\ty}{\tilde{y}}
\newcommand{\tP}{\tilde{P}}

\newcommand{\cqed}{\hfill$\diamond$}

\DeclarePairedDelimiter{\abs}{\lvert}{\rvert}

\newcommand{\cocl}[1]{\ensuremath{\operatorname{#1}}}
\newcommand{\W}[1]{\cocl{W[#1]}}

\newcommand{\NP}{\cocl{NP}}

\newcommand{\fpt}{fixed-parameter tractable\xspace}

\newcommand{\lifetime}{\ensuremath{T}}
\newcommand{\TG}{\mathcal{G}}
\newcommand{\calG}{\mathcal{G}}
\newcommand{\TGcompact}{\ensuremath{\TG=(V, (E_t)_{t\in[\lifetime]})}}

\newcommand{\calW}{\mathcal{S}}

\newcommand{\TempDisjointWalks}{\textsc{Temporally Disjoint Walks}\xspace}
\newcommand{\tdw}{\textsc{Temporally Disjoint Walks}\xspace}
\newcommand{\TempDisjointPaths}{\textsc{Temporally Disjoint Paths}\xspace}
\newcommand{\tdp}{\textsc{Temporally Disjoint Paths}\xspace}

\newcommand{\AllProbs}{\textsc{Temporally Disjoint (Paths/Walks)}}  

\newcommand{\wilog}{without loss of generality}

\newcommand{\problemdef}[3]{
	\begin{center}\fbox{
	\begin{minipage}{0.95\textwidth}
		\noindent
		#1
		\vspace{5pt}\\
		\setlength{\tabcolsep}{3pt}
		\begin{tabularx}{\textwidth}{@{}lX@{}}
			\textrm{Input:}     & #2 \\
			\textrm{Question:}  & #3
		\end{tabularx}
	\end{minipage}}
	\end{center}
}

\title{In Which Graph Structures Can We Efficiently Find Temporally Disjoint Paths and Walks?}

\author{Pascal~Kunz\thanks{Supported by the DFG, Research Training Group 2434 ``Facets of Complexity''.}}
\author{Hendrik~Molter\thanks{Supported by the ISF, grant No.~1456/18, and European Research Council, grant number 949707.}}
\author{Meirav~Zehavi}

\affil{\small Humboldt-Universit\"at zu Berlin, Algorithm Engineering, Berlin, Germany\\ Technische Universit\"at Berlin, Algorithmics and Computational Complexity, Berlin, Germany\\ \texttt{p.kunz.1@tu-berlin.de}}
\affil{\small Department of Computer Science, Ben-Gurion~University~of~the~Negev, 
Beer-Sheva, 
Israel\\ \texttt{molterh@post.bgu.ac.il, meiravze@bgu.ac.il}}

\date{}

\begin{document}

\maketitle

\begin{abstract}
A temporal graph has an edge set that may change over discrete time steps, and a temporal path (or walk) must traverse edges that appear at increasing time steps. 
Accordingly, two temporal paths (or walks) are temporally disjoint if they do not visit any vertex at the same time.
The study of the computational complexity of finding \emph{temporally disjoint} paths or walks in temporal graphs has recently been initiated by Klobas et al.~[IJCAI~'21]. 
This problem is motivated by applications in multi-agent path finding (MAPF), which include robotics, warehouse management, aircraft management, and traffic routing.

We extend Klobas et al.'s research by providing parameterized hardness results for very restricted cases, with a focus on structural parameters of the so-called underlying graph. 
On the positive side, we identify sufficiently simple cases where we can solve the problem efficiently. 
Our results reveal some surprising differences between the ``path version'' and the ``walk version'' (where vertices may be visited multiple times) of the problem, and answer several open questions posed by Klobas et al.
\end{abstract}

\section{Introduction}
Deciding whether a set of vertex pairs (called source-sink pairs) in a graph can be connected by pairwise vertex disjoint paths is a problem that is of fundamental interest in algorithmic graph theory. It was among the first problems that were shown to be NP-complete~\cite{karp1975computational} and  
the further study of the problem is closely tied to one of the most ground-breaking achievements in discrete mathematics in recent history, graph minor theory~\cite{robertson1985graph,robertson1995graph}.
The disjoint path problem is known to be solvable in quadratic time if the number of vertex pairs that need to be connected is constant, that is, the problem is fixed-parameter tractable for the number of sought paths~\cite{KAWARABAYASHI2012424}. On directed graphs, finding two disjoint paths is already NP-hard~\cite{fortune1980directed}, but on directed acyclic graphs the problem is solvable in polynomial time if the number of paths is a constant~\cite{slivkins2010parameterized}.

\citet{KlobasMMNZ23} recently introduced and studied two natural \emph{temporal} versions of the disjoint path problem, called \TempDisjointPaths\ and \TempDisjointWalks. Informally speaking, a temporal graph has an edge set that may change over discrete time steps. Accordingly, temporal paths must traverse edges that appear at increasing time steps and may visit each vertex at most once, whereas a temporal walk may visit each vertex multiple times.
Further, two temporal paths (or walks) are \emph{temporally disjoint} if they do not occupy any vertex at the same time. 
So, analogously to the non-temporal setting, the goal is to find temporal paths (or walks) connecting vertex pairs of a given multiset such that those paths (or walks) are pairwise temporally disjoint.
We give a formal definition in \cref{sec:prelims}.

Due to the asymmetric and non-transitive nature of connectivity in temporal graphs, path-finding related problems behave quite differently in the temporal setting than in the static setting. 
In fact, there are many natural temporal path-finding problems that do not have a direct analogue in the static setting~\cite{CHMZ21,FMNR22a}.
For \TempDisjointPaths\ and \TempDisjointWalks\ the situation is similar. Among other results, \citet{KlobasMMNZ23}, for example, showed that both problems are NP-hard if the underlying graph\footnote{The \emph{underlying graph} of a temporal graph is the static graph containing all edges that appear at least once in the temporal graph.} is a path, a setting where the static disjoint path problem is trivial. Furthermore, they revealed somewhat surprising differences in the computational complexity of \TempDisjointPaths\ and \TempDisjointWalks.
We build on the work of \citet{KlobasMMNZ23} and continue the systematic study of the (parameterized) computational complexity of \TempDisjointPaths\ and \TempDisjointWalks. We provide several new hardness and algorithmic results that expose further interesting differences of the two problem variants and that resolve some of the open questions by \citet{KlobasMMNZ23}.

One of the main application areas for the temporal disjoint path problems is multi-agent path finding (MAPF), an area that has attracted a lot of research from the AI and robotics community in recent years~\cite{Stern19,SternSFK0WLA0KB19,salzman2020research}. 
The goal here is to find paths for multiple agents with the property that all agents can follow these paths concurrently without colliding.
The main difference between classical disjoint path problems and the basic setting of multi-agent path finding problems is that in the latter, we assume the agents move along the paths one step at a time and only collide when they move to the same vertex at the same time. This means that the paths in a solution to a MAPF problem are not necessarily vertex disjoint, but if two paths have a common vertex, that vertex cannot be at the same ordinal position in both paths. A key difference between the classical MAPF settings and the problems we study in this work is the assumption that the network or graph structure may change over time while the agents move through it. 
Real-world applications of MAPF include autonomous vehicles, robotics, automated warehouses, and airport towing~\cite{Stern19,SternSFK0WLA0KB19,salzman2020research}, where 
predictable changes over time in the network topology are well-motivated in many real-world scenarios~\cite{LVM18,Mic16,HS19}.

\subparagraph{Related Work.} The (non-temporal) disjoint path problem is one of the most central problems in algorithmic graph theory, and hence has been extensively studied in the literature for the past few decades. For an overview, we refer to~\citet{korte1990paths}. 

In recent years there has also been intensive research on (non-temporal) multi-agent path finding problems (MAPF) in multiple variations, mostly in the AI and robotics community~\cite{Stern19,SternSFK0WLA0KB19,salzman2020research}. 
MAPF can hence be seen as a generalization of so-called \emph{pebble motion problems} on graphs and it is known to be NP-hard~\cite{goldreich2011finding,yu2013structure}. To the best of our knowledge, the current state of the art (optimal) algorithms for MAPF problems employ the so-called conflict-based search approach~\cite{sharon2015conflict}.

In the area of MAPF, various different settings and problem variations have been considered, including cooperative settings~\cite{Standley10}, robustness requirements~\cite{AtzmonsSFWBZ20} or presence of delays~\cite{ma2017multi}, online settings~\cite{vsvancara2019online}, continuous time settings~\cite{andreychuk2022multi}, explainability requirements~\cite{AlmagorL20},
settings with large agents (that may occupy multiple vertices)~\cite{li2019multi}, and many more. For a more extensive overview, we refer to~\citet{Stern19} and~\citet{SternSFK0WLA0KB19}.


\citet{KlobasMMNZ23} started the study of \TempDisjointPaths\ and \TempDisjointWalks, which can be interpreted as MAPF settings where the availability of edges in the graph may change while the agents are moving. Using the MAPF terminology of \citet{SternSFK0WLA0KB19}, the problem setting 
considers both so-called ``vertex-conflicts'' and ``edge-conflicts'' between agents, that is, two agents cannot occupy the same vertex at the same time and cannot traverse the same edge at the same time. Furthermore, it uses the ``disappear at target'' assumption, that is, once an agent reaches the target vertex, another agent can occupy this vertex again.
We point out that \citet{KlobasMMNZ23} consider so-called \emph{non-strict} temporal paths, that traverse edges that appear at non-decreasing time steps. In this work, we consider so-called \emph{strict} temporal paths that, as described earlier, must traverse edges that appear at (strictly) increasing time steps. The latter models the common assumption in MAPF, that agents can move along at most one edge at each time step~\cite{Stern19}. A closer inspection of the proofs by \citet{KlobasMMNZ23} reveals that the results we mention in the following also hold for the strict case.
\TempDisjointPaths\ is NP-hard even for two source-sink pairs, whereas \TempDisjointWalks\ is W[1]-hard for the number of source-sink pairs but can be solved in polynomial time for a constant number of source-sink pairs~\cite{KlobasMMNZ23}. Furthermore, both problem variants are NP-hard even if the underlying graph is a path~\cite{KlobasMMNZ23}. Nevertheless, \TempDisjointPaths\ is fixed-parameter tractable with respect to the number of source-sink pairs if the underlying graph is a forest~\cite{KlobasMMNZ23}.

Finally, we remark that a different version of disjoint paths in temporal graphs has been studied by \citet{KKK02}. They consider two temporal paths to be disjoint if they to not visit a common vertex, even if that vertex is not occupied by the two temporal paths at the same time. They show that finding two such paths in NP-hard.

\subparagraph{Our Contribution.} The goal of our work is to further understand which structures of the underlying graph can be exploited to solve \TempDisjointPaths\ and \TempDisjointWalks\ efficiently (in terms of parameterized computational complexity). Our first main computational hardness result shows that presumably, we cannot solve \TempDisjointPaths\ and \TempDisjointWalks\ efficiently even in the very restricted case where the number of vertices in the graph is small and the underlying graph is a star (a center vertex that is connected to leaves).
\begin{itemize}
    \item \TempDisjointPaths\ and \TempDisjointWalks\ are NP-hard and W[1]-hard for the number of vertices even if the underlying graph is a star.
\end{itemize}

Recall that we have a multiset of source-sink pairs, that is, the number of source-sink pairs in the input may be much larger than the number of vertices.
This leads us to focusing on cases where we consider the number of source-sink pairs as (part of) the parameter. As mentioned before, \citet{KlobasMMNZ23} showed that \TempDisjointPaths\ is fixed-parameter tractable with respect to the number of source-sink pairs if the underlying graph is a forest. They left open whether this algorithm can be generalized to an FPT-algorithm where the parameter is the number of source-sink pairs combined with some distance-to-forest measure for the underlying graph. They also left open whether a similar algorithm can be found for \TempDisjointWalks. We resolve both of these open questions.

For \TempDisjointPaths\ we show that we presumably cannot obtain an FPT-algorithm even if we combine the number of source-sink pairs with the vertex cover number of the underlying graph.
\begin{itemize}
    \item \TempDisjointPaths\ is W[1]-hard for the combination of the number of source-sink pairs and the vertex cover number of the underlying graph.
\end{itemize}
This result excludes several popular distance-to-forest measures as potential parameters, such as the treewidth or the feedback vertex number, which are smaller than the vertex cover number. On the positive side, we can show that we can use the feedback edge number (which is a distance-to-forest measure that is incomparable to the vertex cover number) as an additional parameter to obtain tractability.
\begin{itemize}
    \item \TempDisjointPaths\ is fixed-parameter tractable with respect to the combination of the number of source-sink pairs and the feedback edge number of the underlying graph.
\end{itemize}

The parameterized complexity of \TempDisjointWalks\ is surprisingly different. For this problem we can show that, in contrast to the path version, we presumably cannot obtain an FPT-algorithm for the number of source-sink pairs even if the underlying graph is a star.
\begin{itemize}
    \item \TempDisjointWalks\ is W[1]-hard for the number of source-sink pairs even if the underlying graph is a star.
\end{itemize}
This is quite surprising given that \citet{KlobasMMNZ23} showed that \TempDisjointWalks\ is easier to solve than \TempDisjointPaths\ when the number of source-sink pairs is considered as a parameter and the underlying graphs is unrestricted.
As mentioned before, they showed that in general, \TempDisjointWalks\ can be solved in polynomial time if the number of source-sink pairs is constant whereas \TempDisjointPaths\ is NP-hard already for two source-sink pairs.
On the positive side, if the underlying graph is restricted to be a path, we can achieve fixed-parameter tractability for \TempDisjointWalks.
\begin{itemize}
    \item \TempDisjointWalks\ is fixed-parameter tractable with respect to the number of source-sink pairs if the underlying graph is a path.
\end{itemize}

Our results provide a quite complete picture of the parameterized complexity of \TempDisjointPaths\ and \TempDisjointWalks\ when the number of source-sink pairs and structural parameters (particularly ones that measure similarity to forests) of the underlying graph are considered. We point out remaining open cases and future research directions in \cref{sec:conclusion}.

\section{Preliminaries and Problem Definition}\label{sec:prelims}
We denote by~$\mathbb N$ and~$\mathbb N_0$ the natural numbers excluding and including~$0$, respectively.
An interval on $\mathbb N_0$ from $a$ to $b$ is denoted by $[a,b] \coloneqq \{ i\in \mathbb N_0 \mid a \leq i \leq b\}$
and $[a] \coloneqq [1,a]$.
A \emph{static} undirected graph~$G=(V,E)$ consists of a vertex set~$V$ and an edge set~$E\subseteq \binom{V}{2}$.
A \emph{$(s,z)$-walk} (or \emph{walk} from $s$ to $z$) in~$G$ of length~$k$ from vertex $s=v_0$ to vertex $z=v_k$ is a sequence
$P = \left(v_{i-1},v_i\right)_{i=1}^k$
of \emph{static transitions} 
such that for all $i\in[k]$ we have that $\{v_{i-1},v_i\}\in E$. 
The $(s,z)$-walk $P$ 
is called a \emph{$(s,z)$-path} (or \emph{path} from $s$ to $z$) if~$v_i\neq v_j$ whenever $i\neq j$.

A \emph{temporal graph} $\calG = (V,E_1,\ldots,E_T)$ or $\TGcompact$ consists of a vertex set $V$ and $T$ edge sets $E_1,\ldots,E_T \subseteq \binom{V}{2}$.
The \emph{underlying graph} of $\calG$ is the static graph $G_U=(V,E_U)$ with $E_U \coloneqq \bigcup_{i =1}^T E_i$.
The indices $1,\ldots,T$ are the \emph{time steps} of $\calG$.
The temporal graph $\calG$ is a \emph{temporal tree}, \emph{star}, or \emph{line}, respectively, if its underlying graph is a tree, star, or path.
We call the pair $(e,i)$ a \emph{time edge} of $\TG$ if $e \in E_i$.
The graph~$(V,E_i)$ is the \emph{$i$-th~layer} of~$\TG$.
We will say that a vertex $v \in V$ is \emph{isolated} in a layer $i \in [T]$ if it is not incident to any edges in $E_i$.
It is active \emph{active} in the time step $i \in [T]$, if there are (not necessarily distinct) layers $j,j' \in [T]$ such that $j \leq i \leq j'$ and $v$ is not isolated in $E_j$ or $E_{j'}$.

A \emph{temporal $(s,z)$-walk} (or \emph{temporal walk} from $s$ to $z$) in~$\TG$ of length~$k$ from vertex $s=v_0$ to vertex $z=v_k$ is a sequence
$P = \left(\left(v_{i-1},v_i,t_i\right)\right)_{i=1}^k$
of \emph{transitions} 
such that for all $i\in[k]$ we have that $\{v_{i-1},v_i\}\in E_{t_i}$ and for
all $i\in [k-1]$ we have that $t_i < t_{i+1}$. 
The temporal ($s,z$)-walk $P$ 
is called a \emph{temporal $(s,z)$-path} (or \emph{temporal path} from $s$ to $z$) if~$v_i\neq v_j$ whenever $i\neq j$.
%
%
The \emph{arrival time} of $P$ is $t_k$. 
We say that $P$ 
\emph{visits} the vertices $V(P)\coloneqq \{ v_i \mid i\in [0,k] \}$ in order $v_0, v_1, \ldots, v_k$. We say that $P$ \emph{occupies} vertex $v_i$ during the time 
interval $[t_{i}, t_{i+1}]$, for all $i \in [k-1]$.
Furthermore, we say that $P$ occupies $v_0$ during time interval $[t_1,t_1]$ and
$P$ occupies~$v_k$ during time interval $[t_k,t_k]$.
We say that $P$ \emph{follows} the (static) path or walk $P'=\left(v_{i-1},v_i\right)_{i=1}^k$ of the underlying graph of $\TG$.
If the arrival time of $P$ is the smallest possible among all temporal $(s,z)$-walks in $\TG$, we call $P$ \emph{foremost}. If for all $v_i$ with $i\in [k]$ we have that $t_i$ is the arrival time of a foremost $(s,v_i)$-path in $\TG$, then we call $P$ \emph{prefix foremost}. If there is a temporal $(s,z)$-path in $\TG$, then there also always exists a prefix-foremost $(s,z)$-path and such a path can be computed in polynomial time~\cite{WuCKHHW16}.
Given two temporal walks $P_1, P_2$ 
we say that $P_1$ and $P_2$ \emph{temporally intersect} 
if there exists a vertex $v$ and two time intervals $[a_1,b_1],[a_2,b_2]$, where $[a_1,b_1]\cap[a_2,b_2]\neq \emptyset$, 
such that $v$ is occupied by $P_1$ during~$[a_1,b_1]$ and by~$P_2$ during~$[a_2,b_2]$. We say that $P_1$ and $P_2$ \emph{temporally disjoint} if they are not temporally intersecting. 
The problem \TempDisjointPaths is formally defined as follows.

\medskip

\problemdef{\TempDisjointPaths}{A temporal graph~$\TGcompact$ 
	and a multiset~$S$ of
	source-sink pairs containing elements from $V\times V$.}
{Are there pairwise
	temporally disjoint 
	temporal $(s_i,z_i)$-paths for all~$(s_i,z_i)\in S$?}

\medskip

The problem \TempDisjointWalks{} receives the same input but asks whether
there are pairwise
temporally disjoint 
temporal $(s_i,z_i)$-walks for all~$(s_i,z_i)\in S$. 
Given an instance of \AllProbs, we use $\hat{S}$ to denote the set of vertices in $V$ that appear as sources or sinks in $S$, that is, $\hat{S}=\{s\in V\mid (s,z)\in S\text{ for some }z\in V\}\cup \{z\in V\mid (s,z)\in S\text{ for some }s\in V\}$.

We study the (parameterized) computational complexity of the two problems introduced  above. We use the following standard concepts from parameterized complexity theory~\cite{DF13,FG06,Cyg+15}.
A \emph{parameterized problem}~$L\subseteq \{(x,k)\in \Sigma^*\times \mathbb N\}$ is a subset of all instances~$(x,k)$ from~$\Sigma^*\times \mathbb N$, where~$k$ denotes the \emph{parameter}.
A parameterized problem~$L$ is 
in the class FPT (or \emph{fixed-parameter tractable}) if there is an algorithm that decides every instance~$(x,k)$ for~$L$ in~$f(k)\cdot |x|^{O(1)}$ time, 
where~$f$ is any computable function that depends only on the parameter. If a parameterized problem $L$ is W[1]-hard, then it is presumably not fixed-parameter tractable~\cite{DF13,FG06,Cyg+15}.

\section{Parameterized Hardness of \AllProbs}\label{sec:hardness}

In this section, we analyze the parameterized hardness of \AllProbs\ with respect to structural parameters of the underlying graph (combined with the number of source-sink pairs).
We first consider the number $|V|$ of vertices in the input temporal graph as a parameter, which one might consider as the largest structural parameter of the underlying graph.
In \cref{sec:hardnesspaths} we show that \TempDisjointPaths\ is W[1]-hard when parameterized by the combination of the vertex cover number of the underlying graph and the number of source-sink pairs. In \cref{sec:hardnesswalks} we show that \TempDisjointWalks\ is W[1]-hard when parameterized by the number of source-sink pairs even if the underlying graph is a star. 

We first show that \AllProbs\ is NP-hard \W{1}-hard when parameterized by $|V|$ and that this hardness holds even on temporal stars. 

\begin{theorem}
	\label{thm:n-star}
	\AllProbs\ on temporal stars is NP-hard and \W{1}-hard with respect to the number~$|V|$ of vertices.
\end{theorem}

We will prove \cref{thm:n-star} for \TempDisjointWalks\ by a parameterized reduction from the \W{1}-hard problem \textsc{Unary Bin Packing} parameterized by the number of bins~\cite{jansen2013bin}. In our reduction, all temporal walks in any solution will be paths, hence the result also holds for \TempDisjointPaths.

In \textsc{Unary Bin Packing}, the input consists of $m$ items of sizes $x_1,\ldots,x_m \in \N$, a number of bins $b \in \N$, and a bin size $B \in \N$ with all integers encoded in unary.
One is asked to decide whether there is an assignment $f \colon [m] \to [b]$ of items to bins such that $\sum_{i \in f^{-1}(j)} x_i \leq B$ for all $j \in [b]$, that is, the total size of the items assigned to any bin is at most $B$.
The parameter is~$b$.
We may assume \wilog{} that $S \coloneqq \sum_{i=1}^m x_i = b  B$.
If $S > b  B$, then the instance is a clearly a \no-instance.
If $S < b  B$, then one can add $b B - S$ items of size~$1$ without changing whether or not the instance is a \yes-instance.
Intuitively, the advantage of this restriction is that the problem can then be thought of as trying to fill every bin completely.

The idea behind the reduction is as follows.
For each bin $j \in [b]$, there are $B$ copies of a terminal pair representing that bin (and there is also an additional ``dummy'' terminal pair).
The reduction outputs a temporal graph that, for each item $i\in [m]$, contains a sequence of consecutive layers representing that item.
The gadget representing item $i$ is constructed in such a way that the following holds:
in order to be able to send a sufficient number of pairwise temporally disjoint temporal walks between each terminal pair, one is forced to send exactly $x_i$ temporal walks between the same terminal pair in the gadget for item $i$.

We will now give the formal construction for the parameterized reduction from \textsc{Unary Bin Packing} parameterized by $b$ to \tdw parameterized by $|V|$.

\begin{construction}
	\label{constr:n-star}
	Let $(x_1,\ldots,x_m,b,B)$ be an instance of \textsc{Unary Bin Packing} with the aforementioned restriction.
	We will define an instance $(\calG,S)$ of \tdw.
	The vertex set of $\calG$ is $V \coloneqq \{s_1,\ldots,s_b,z_1,\ldots,z_b,\tilde{s},\tilde{z},c\}$.
	Hence, $\abs{V} = 2b + 3$.
	We let the multiset~$S$ contain $B$ pairs $(s_j,z_j)$ for every $j \in [b]$ and $m(b-1)$ pairs $(\ts,\ttt)$.
	We will now define the layers of~$\calG$.
	For every item $i \in [m]$, we add a sequence of layers $E^i_1,\ldots,E^i_{2bx_i}$.
	We let 
	\begin{align*}
		E^i_{2j-1} & \coloneqq \{\{\ts,c\}\} \text{ for all } j \in [b],\\
		E^i_{2j} & \coloneqq \{\{c,\ttt\},\{s_j,c\}\} \text{ for all } j \in [b],\\
		E^{i}_{2bj'+j} & \coloneqq \{\{c,z_{j}\}\} \text{ for all } j \in [b], j' \in [x_i -1],\\
		E^{i}_{2bj'+j+1} & \coloneqq \{ \{s_j, c\}\} \text{ for all } j \in [b], j' \in [x_i -1], \\
		E^i_{2bx_i + 2j-1} & \coloneqq \{\{c,z_j\},\{\ts,c\}\} \text{ for all } j \in [b], \text{ and }\\
		E^i_{2bx_i + 2j} & \coloneqq \{\{c,\ttt\}\} \text{ for all } j \in [b].
	\end{align*}
	The layers corresponding to an item are illustrated in \Cref{fig:n-star}.
	We place the layers for the first item first, followed by the layers for the second item, and so on.
	\cqed
\end{construction}

\begin{figure}[t]
	\centering
	\begin{tikzpicture}[scale=.7]
	\tikzpreamble
	\def\nsc{0.5}
	
	\node[xnode] (s1) at (0,0) {};
	\node[below] at (s1) {$s_1$};
	\node[xnode] (s2) at (3,0) {};
	\node[below] at (s2) {$s_2$};
	\node[xnode] (s3) at (6,0) {};
	\node[below] at (s3) {$s_3$};
	\node () at (9,0) {\Huge $\ldots$};
	\node[xnode] (sb) at (12,0) {};
	\node[below] at (sb) {$s_b$};
	
	\def\top{8}
	\node[xnode] (t1) at (0,\top) {};
	\node[above] at (t1) {$z_1$};
	\node[xnode] (t2) at (3,\top) {};
	\node[above] at (t2) {$z_2$};
	\node[xnode] (t3) at (6,\top) {};
	\node[above] at (t3) {$z_3$};
	\node () at (9,\top) {\Huge $\ldots$};
	\node[xnode] (tb) at (12,\top) {};
	\node[above] at (tb) {$z_b$};
	
	\def\mid{4}
	\node[vnode] (c) at (6,\mid) {};
	\node[xnode] (st) at (-1,\mid) {};
	\node[left] at (st) {$\ts$};
	\node[xnode] (tt) at (13,\mid) {};
	\node[right] at (tt) {$\ttt$};
	
	\draw (s1) to node[midway,align=center,sloped,font=\scriptsize,yshift=4.25pt]{ $2b(j'-1) +2$}(c);
	\draw (s2) to node[midway,align=center,sloped,font=\scriptsize,yshift=4.25pt]{ $2b(j'-1) +4$}(c);
	\draw (s3) to node[midway,align=center,sloped,font=\scriptsize,yshift=4.25pt]{ $2b(j'-1) +6$}(c);
	\draw (sb) to node[midway,align=center,sloped,font=\scriptsize,yshift=4.25pt]{ $2bj'$}(c);
	\draw (t1) to node[midway,align=center,sloped,font=\scriptsize,yshift=4.25pt]{ 
	$2bj' +1$
	}(c);
	\draw (t2) to node[midway,align=center,sloped,font=\scriptsize,yshift=4.25pt]{ $2bj' +3$}(c);
	\draw (t3) to node[midway,align=center,sloped,font=\scriptsize,yshift=4.25pt]{ $2bj' +5$}(c);
	\draw (tb) to node[midway,align=center,sloped,font=\scriptsize,yshift=4.25pt]{ $2b(j'+1)-1$}(c);
	\draw (st) to node[midway,align=center,sloped,font=\scriptsize,yshift=4.25pt]{ 
	$2j -1,2bx_i+2j -1$
	}(c);
	\draw (tt) to node[midway,align=center,sloped,font=\scriptsize,yshift=4.25pt]{ $2j,2bx_i+2j$}(c);
	
	\node[align=center] () at (6,-1) {$j\in [b], j' \in [x_i]$};
	\end{tikzpicture}
	\caption{Illustration of \cref{constr:n-star}: Labels on the edges indicate the indices $j$ such that the edge is contained in $E^i_j$, where $j$ and $j'$ range over the values indicated at the bottom.
	The central vertex $c$ is represented by the square.}
	\label{fig:n-star}
\end{figure}
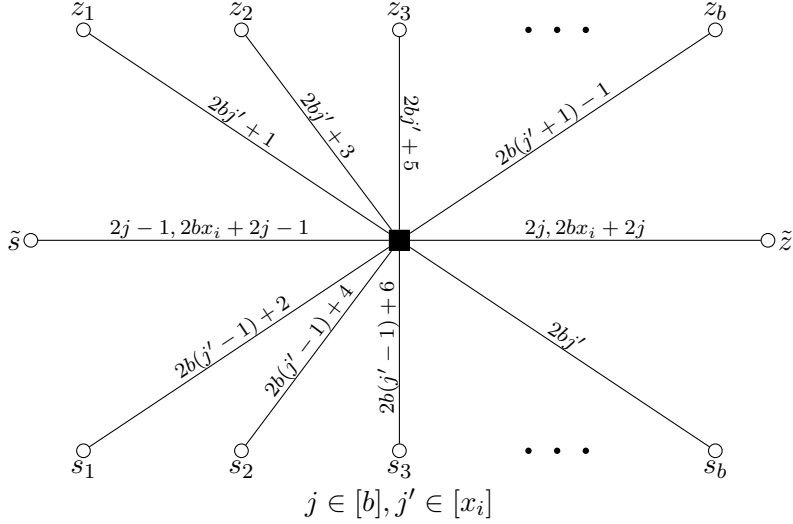

With this, we now prove \cref{thm:n-star}.

\begin{proof}[Proof of \cref{thm:n-star}]
	The proof is based on the reduction described in \cref{constr:n-star}.
	Clearly, \cref{constr:n-star} can be computed in polynomial time and the number of vertices in the output instance is bounded by a function in the parameter $b$.
	
	It remains to show that the input instance $(x_1,\ldots,x_m,b,B)$ for \textsc{Unary Bin Packing} is a \yes-instance if and only if the output instance $(\calG,S)$ for \tdw{} is a \yes-instance.
	
	\RD{} Suppose that $f\colon [m] \to [b]$ is an assignment such that $\sum_{i \in f^{-1}(j)} x_i \leq B$ for all $j \in [b]$.
	Because $S \coloneqq \sum_{i=1}^m x_i = b B$ and $\sum_{i \in f^{-1}(j)} x_i \leq B$ for all $j \in [b]$, it follows that  $\sum_{i \in f^{-1}(j)} x_i = B$ for all $j \in [b]$.
	We create a set $\calW$ of temporally disjoint $S$-walks in $\calG$ as follows.
	For each item $i \in [m]$, we add $x_i$ walks from $s_{f(i)}$ to $z_{f(i)}$ as well as $b-1$ walks from $\ts$ to $\ttt$.
	For $j\in [x_i]$, the $j$-th $(s_{f(i)},z_{f(i)})$-walk moves from $s_{f(i)}$ to $c$ in layer $E^i_{2b(j-1)+2f(i)}$ and from $c$ to  $z_{f(i)}$ in layer $E^i_{2bj+2f(i)-1}$.
	When it comes to the $b-1$ walks from $\ts$ to $\ttt$, the first $f(i)-1$ of these walks move from $\ts$ to $c$ in layers $E^i_{1},E^i_{3},\ldots,E^i_{2f(i)-1}$ and from $c$ to $\ttt$ in $E^i_{2},E^i_{4},\ldots,E^i_{2f(i)}$.
	The remaining $b-f(i)$ walks from $\ts$ to $\ttt$, take place in the layers $E^i_{2bx_i+2f(i)+1},\ldots,E^i_{2x_ib}$.
	In all, this means that the number of $(s_j,z_j)$-walks is $\displaystyle \sum_{\substack{i \in [m] \\ f(i) = j }} x_i = B$ and the number of $(\ts,\ttt)$-walks is $m(b-1)$.
	
	\LD{} Suppose that $\calW$ is a set of temporally disjoint walks for $S$ in $\calG$.
	First, consider any item $i \in [m]$ and the corresponding layers $E^i_1,\ldots,E^i_{2bx_i}$.
	At most $x_i$ paths may move from any $s_j$ to $z_j$ in those layers, because each of the edges $\{s_j,c\}$ and $\{c,z_j\}$ is only present $x_i$ times.
	Since a total of $bB = \sum_{i=1}^m x_i$ walks must move from some $s_j$ to $z_j$, it follows that exactly $x_i$ paths must go from any $s_j$ to $z_j$ in the layers corresponding to each item.
	Then, for each item at most $b-1$ walks can move from $\ts$ to $\ttt$ and, therefore, exactly $b-1$ walks must do this.
	The only way to achieve this is if the layers corresponding to a particular item only contain $(s_j,z_j)$-paths for a single index $j$.
	Let $f(i)$ be that index for item $i$.
	Since only $B$ paths from any particular $s_j$ to $z_j$ are needed, it follows that $\sum_{i \in f^{-1}(j)} x_j \leq B$.
	Hence, $f$ is an assignment of items to bins with the required properties.
\end{proof}

\subsection{Parameterized Hardness of \TempDisjointPaths}\label{sec:hardnesspaths}

In this section, we show that \TempDisjointPaths\ is W[1]-hard when parameterized by the combination of source-sink pairs and the vertex cover number of the underlying graph. This shows to which extend we can expect to generalize the FPT-algorithm for \TempDisjointPaths\ for the number of source-sink pairs on temporal forests by \citet{KlobasMMNZ23}. Our result rules out FPT-algorithms for \TempDisjointPaths\ parameterized by the number of source-sink pairs combined with e.g.\ the treewidth or the feedback vertex number of the underlying graph, since both of theses parameters are smaller than the vertex cover number. To obtain tractability, we have to use parameters that are larger or incomparable to the vertex cover number, such as the feedback edge number. In \cref{sec:FESalgo} we show that this is indeed possible.

\begin{theorem}\label{thm:w1hardnessVCN}
\TempDisjointPaths\ is W[1]-hard when parameterized by the combination of the number $|S|$ of source-sink pairs and the vertex cover number of the underlying graph.
\end{theorem}

We will prove \cref{thm:w1hardnessVCN} by a parameterized reduction from the W[1]-hard problem \textsc{Multicolored Clique} parameterized by the number of colors~\cite{fellows2009multipleinterval}.  Here, given a $k$-partite graph $G=(V_1\uplus V_2 \uplus\ldots\uplus V_k, E)$, we are asked whether $G$ contains a clique of size $k$. The parameter is~$k$. If $v\in V_i$, then we say that $v$ has \emph{color} $i$. W.l.o.g.\ we assume that $|V_1|=|V_2|=\ldots=|V_k|=n$ and that every vertex has at least one neighbor of every color. Let $E_{i,j}$ denote the set of all edges between vertices from $V_i$ and $V_j$. We assume w.l.o.g.\ that $|E_{i,j}|=m$ for all $i\neq j$. 



Before we give the complete description on how we construct an instance of \TempDisjointPaths\
we define the following gadget $H(p,q,t)$. We will use this gadget for ``vertex selection'' as well as ``edge selection''. Each gadget has one source-sink pair associated with it and each gadget has three parameters: $p$, $q$, and $t$. Intuitively speaking, $p$ determines how many options there are for a temporal path from the source to the sink of the gadget. Parameter $q$ determines how many other temporal paths can traverse the gadget without temporally intersecting each other and the temporal path corresponding to the source-sink pair of the gadget. Finally, $t$ is a temporal offset.

Informally speaking, when we use $H(p,q,t)$ as a vertex selection gadget for some color $i$, we will set $p$ to the number of vertices of color $i$. Hence, choosing a temporal path from the source to the sink of the gadget corresponds to selecting a vertex of color $i$. We will set $q$ to $k-1$, that is, for each color \emph{different} from $i$, one temporal path can traverse the gadget. Intuitively, each of those temporal paths will verify that the vertex we selected for color $i$ has an edge to each vertex we selected for the other colors. 

When we use $H(p,q,t)$ as an edge selection gadget for color combination $i,j$, we will set $p$ to the number of edges between vertices of color $i$ and vertices of color $j$. Hence, choosing a temporal path from the source to the sink of the gadget corresponds to selecting an edge of color combination $i,j$. We will set $q$ to one, that is, only one other temporal path can traverse the gadget. This will be the temporal path verifying that the vertex we selected for color $i$ has an edge (the one selected in this gadget) to the vertex we selected for color $j$.

Formally, the gadget is defined as follows.

\begin{construction}
We construct a temporal graph $H(p,q,t)$ with one source and one sink vertex as follows.
\begin{itemize}
    \item The ``skeleton'' of the gadget consists of four vertices $s,z,c,c'$ and two edges $\{s,c\}$, $\{c',z\}$.
    \item The ``body'' of the gadget consists of $p$ different $(q+1)$-tuples of vertices, $(w^{(\ell)}_0, w^{(\ell)}_1, w^{(\ell)}_2, \ldots, w^{(\ell)}_q)$ for $\ell\in[p]$.

    It additionally contains $2p(q+1)$ edges, $\{c,w^{(\ell)}_r\}$, $\{w^{(\ell)}_r,c'\}$ for $\ell\in[p]$ and $0\le r\le q$.
\end{itemize}
The gadget is labelled as follows.
\begin{itemize}
    \item Edge $\{s,c\}$ is labelled with one and $\{c',z\}$ is labelled with $4n^3$ (this will be the largest time label in the constructed temporal graph).
    \item For the $(q+1)$-tuple $(w^{(\ell)}_0, w^{(\ell)}_1, w^{(\ell)}_2, \ldots, w^{(\ell)}_q)$ we add the following labels.
    \begin{itemize}
        \item We label $\{c,w^{(\ell)}_0\}$ with $t+(2\ell-1) q-1$ and we label $\{w^{(\ell)}_0,c'\}$ with $t+2\ell q+2$.
        \item For $r\in[q]$ we label $\{c,w^{(\ell)}_r\}$ with $t+(2\ell-1) q+2r-1$ and we label $\{w^{(\ell)}_r,c'\}$ with $t+(2\ell-1) q+2r$.
    \end{itemize}
\end{itemize}
This finishes the construction of gadget $H(p,q,t)$, it is illustrated in \cref{fig:gadget}. We call $(s,z)$ the source-sink pair of the gadget, we call vertex $c$ the entry of the gadget, and we call vertex $c'$ the exit of the gadget. \cqed
\begin{figure}[t]
\begin{center}
\begin{tikzpicture}[line width=1pt, scale=.8, xscale=2.4]
    \node[vert,label=below:$s$] (S) at (0,1) {};
    \node[vert2,label=left:$c$] (C1) at (1.5,4) {};
    \node[vert2,label=right:$c'$] (C2) at (6.5,4) {};
    \node[vert,label=below:$z$] (Z) at (8,1) {};
    \node[vert,label=above:$w^{(\ell)}_q$] (WQ) at (4,10) {};
    \node[vert,label=above:$w^{(\ell)}_{q-1}$] (WQ2) at (4,8) {};
    \node (D) at (4,7) {$\vdots$};
    \node[vert,label=above:$w^{(\ell)}_{2}$] (W2) at (4,5) {};
    \node[vert,label=above:$w^{(\ell)}_{1}$] (W1) at (4,3) {};
    \node[vert,label=above:$w^{(\ell)}_{0}$] (W0) at (4,0) {};

	\draw (S) --node[timelabel] {$1$} (C1);
	\draw (C2) --node[timelabel] {$4n^3$} (Z);
	\draw (C1) edge[bend left=20] node[timelabel] {$t+2\ell q-1$} (WQ);
	\draw (C1) --node[timelabel] {$t+2\ell q-3$} (WQ2);
	\draw (C1) --node[timelabel] {$t+(2\ell-1) q+3$} (W2);
	\draw (C1) --node[timelabel] {$t+(2\ell-1) q+1$} (W1);
\draw (C1) edge[bend right=20] node[timelabel] {$t+(2\ell-1) q-1$} (W0);
	
    \draw (C2) edge[bend right=20] node[timelabel] {$t+2\ell q$} (WQ);
	\draw (C2) --node[timelabel] {$t+2\ell q-2$} (WQ2);
	\draw (C2) --node[timelabel] {$t+(2\ell-1) q+4$} (W2);
	\draw (C2) --node[timelabel] {$t+(2\ell-1) q+2$} (W1);
	\draw (C2) edge[bend left=20] node[timelabel] {$t+2\ell q+2$} (W0);
\end{tikzpicture}
    \end{center}
    \caption{Illustration of gadget $H(p,q,t)$. Only one $(q+1)$-tuple $(w^{(\ell)}_0, w^{(\ell)}_1, w^{(\ell)}_2, \ldots, w^{(\ell)}_q)$ is depicted. The gray vertices $c,c'$ form a vertex cover of the gadget.}\label{fig:gadget}
\end{figure}
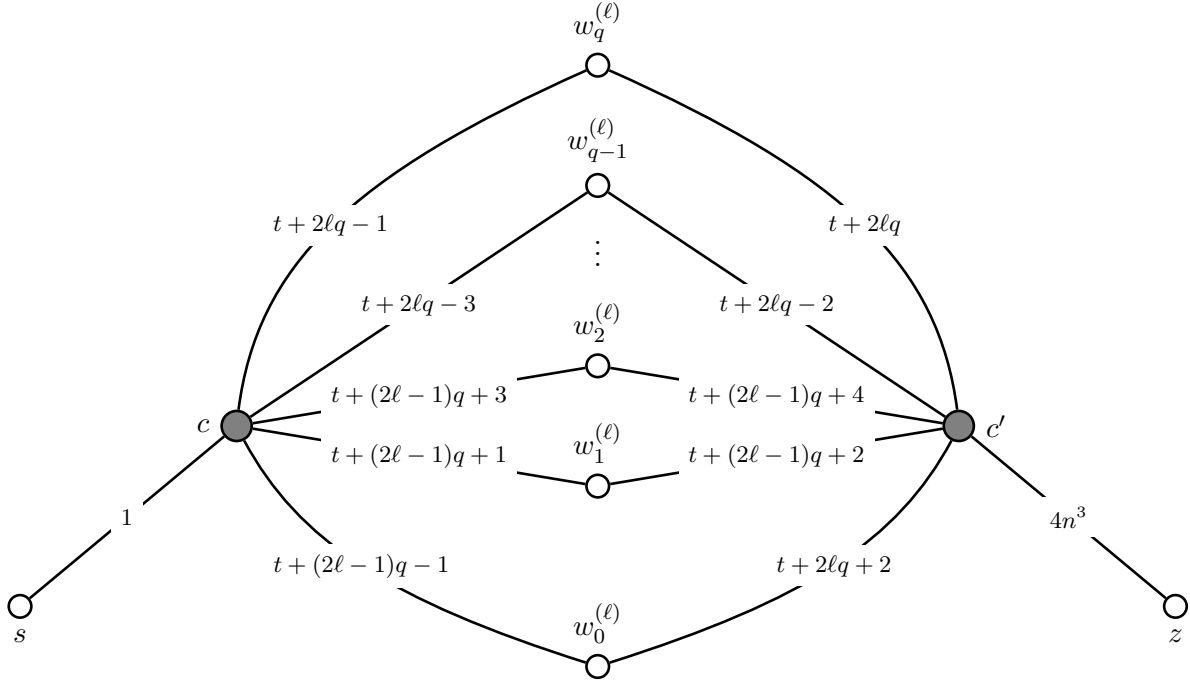
\end{construction}

We first prove some properties of $H(p,q,t)$, reflecting the intuition we gave in the beginning.
\begin{lemma}\label{claim1}
Consider a temporal path $P$ from $s$ to $z$ in gadget $H(p,q,t)$. If $P$ visits $w^{(\ell)}_r$ for some $\ell\in[p]$ and $r\in [q]$ (that is, $r\neq 0$), then every temporal path $P'$ from the entry vertex of the gadget to the exit vertex of the gadget temporally intersects $P$.
\end{lemma}
\begin{proof}
Let $P$ be a temporal path from $s$ to $z$ in gadget $H(p,q,t)$ that visits $w^{(\ell)}_r$ for some $\ell\in[p]$ and $r\in [q]$.
Note that edge $\{s,c\}$ has label one and edge $\{c',z\}$ has label $4n^3$. Edges $\{c,w^{(\ell)}_r\}$ and $\{w^{(\ell)}_r,c'\}$ have consecutive labels for $r\neq 0$, say $x$ and $x+1$ with $1<x<4n^3-1$. Hence, vertex $c$ is occupied by $P$ from time 1 to $x$ and vertex $c'$ is occupied by $P$ from time $x+1$ to $4n^3$.
Now any path $P'$ from $c$ to $c'$ needs to start at $c$ at some time $t>x$. However, then it cannot arrive at $c'$ earlier than $x+1$. Hence, $P'$ is temporally intersecting $P$.
\end{proof}

\begin{lemma}\label{claim2}
Consider a temporal path $P$ from $s$ to $z$ in gadget $H(p,q,t)$ such that $P$ visits $w^{(\ell)}_0$ for some $\ell\in[p]$. Then there are exactly $q$ temporal paths $P'_1,\ldots, P'_q$ from the entry vertex of the gadget to the exit vertex of the gadget that pairwise are not temporally intersecting and do not temporally intersect $P$. Furthermore, each temporal path $P'_1,\ldots,P'_q$ visits one of the vertices $w^{(\ell)}_1,\ldots,w^{(\ell)}_q$.
\end{lemma}
\begin{proof}
If $P$ visits $w^{(\ell)}_0$ for some $\ell\in[p]$, then it occupies the entry vertex $c$ from time 1 to $t+(2\ell-1) q-1$ and $P$ occupies the exit vertex $c'$ from time $t+2\ell q+2$ to $4n^3$. Now let $P'_r$ be a temporal path that starts at $c$, then continues to $w^{(\ell')}_r$, and then arrives at $c'$2. Assume that $P$ does not temporally intersect $P'_r$. 
Note that if $\ell\neq \ell'$, then $P'_r$ temporally intersects $P$ in either $c$ or $c'$. Hence we can conclude that $\ell=\ell'$. Furthermore, if $r=0$, then $P'_r$ and $P$ temporally intersect in $w^{(\ell)}_0$. It follows that $r\neq 0$.
We have that $P'_r$ starts at $c$ (at the latest) at time $t+(2\ell-1) q+2r-1$ and arrives at $c'$ at time $t+(2\ell-1) q+2r$. It is straightforward to check that $P'_r$ and $P'_{r'}$ do not temporally intersect if and only if $r\neq r'$. The lemma follows.
\end{proof}

We now continue with the construction of the temporal graph $\TG$ and the set of source-sink pairs~$S$.
\begin{construction}\label{constr}
Let $G=(V_1\uplus V_2 \uplus\ldots\uplus V_k, E)$ be an instance of \textsc{Multicolored Clique}.
For each color $i$ we create one gadget $H(n,k-1,2(i-1)(kn+m)+2k)$ and add its source-sink pair to $S$. For each color combination~$i,j$ with $i<j$ we create one gadget $H(m,1,2(i-1)(kn+m)+2kn)$ and add its source-sink pair to $S$. 
For each color $i$ let the vertices in $V_i$ be ordered in a fixed but arbitrary way, that is, $V_i=\{v^{(i)}_1, v^{(i)}_2, \ldots, v^{(i)}_n\}$.
For each color combination $i,j$ with $i<j$ let the edges in $E_{i,j}$ be ordered in a fixed but arbitrary way, that is, $E_{i,j}=\{e^{(i,j)}_1, e^{(i,j)}_2, \ldots, e^{(i,j)}_m\}$.
We, additionally, add: two vertices $s_{i,j}$, $z_{i,j}$ and add source-sink pair~$(s_{i,j},z_{i,j})$ to $S$; $2m$ vertices $\{v_e,v_e'\mid e\in E_{i,j}\}$; and $2n$ vertices $\{u_\ell^{(i,j)}\mid v_\ell^{(i)}\in V_i\}\cup \{v_\ell^{(i,j)}\mid v_\ell^{(j)}\in V_j\}$.
We add the following edges.
\begin{itemize}
    \item We add an edge between $s_{i,j}$ and each vertex $u_\ell^{(i,j)}$ (with $v_\ell^{(i)}\in V_i$).
    \item We add an edge between each vertex $u_\ell^{(i,j)}$ (with $v_\ell^{(i)}\in V_i$) and the entry of the gadget corresponding to color~$i$.
    \item We add an edge between $z_{i,j}$ and each vertex $v_\ell^{(i,j)}$ with $v_\ell^{(j)}\in V_j$.
    \item We add an edge between each vertex $v_\ell^{(i,j)}$ (with $v_\ell^{(j)}\in V_j$) and the exit of the gadget corresponding to color~$j$.
    \item For every $e\in E_{i,j}$ we add the following edges.
    \begin{itemize}
        \item We add an edge between $v_e$ and the exit of the gadget corresponding to color $i$.
        \item We add an edge between $v_e$ and the entry of the gadget corresponding to color combination~$i,j$.
        \item We add an edge between $v_e'$ and the exit of the gadget corresponding to color combination~$i,j$.
        \item We add an edge between $v_e'$ and the entry of the gadget corresponding to color $j$.
    \end{itemize}
\end{itemize}

For each edge $e^{(i,j)}_\ell=\{v^{(i)}_{\ell'},v^{(j)}_{\ell''}\}\in E_{i,j}$ with $i<j$ we add the following time labels.
\begin{itemize}
    \item We add label $1$ to the edge between $s_{i,j}$ and $u^{(i,j)}_{\ell'}$.
    \item We add label $2(i-1)(kn+m)+2k+(2\ell'-1)(k-1)+2(j-1)-2$ to the edge between $u^{(i,j)}_{\ell'}$ and the entry of the gadget corresponding to color $i$.
    \item We add label $2(i-1)(kn+m)+2k+(2\ell'-1)(k-1)+2(j-1)+1$ to the edge between $v_e$ and the exit of the gadget corresponding to color $i$.
    \item We add label $2(i-1)(kn+m)+2kn+2\ell-1$ to the edge between $v_e$ and the entry of the gadget corresponding to color combination $i,j$.
    \item We add label $2(i-1)(kn+m)+2kn+2\ell+2$ to the edge between $v_e'$ and the exit of the gadget corresponding to color combination $i,j$.
    \item We add label $2(j-1)(kn+m)+2k+(2\ell''-1)(k-1)+2i-2$ to the edge between $v_e'$ and the entry of the gadget corresponding to color $i$.
    \item We add label $4n^3$ to the edge between $z_{i,j}$ and $v^{(i,j)}_{\ell''}$.
    \item We add label $2(j-1)(kn+m)+2k+(2\ell''-1)(k-1)+2i+1$ to the edge between $v^{(i,j)}_{\ell''}$ and the exit of the gadget corresponding to color $j$.
\end{itemize}
%
This finishes the construction of the \TempDisjointPaths\ instance $(\TGcompact,S)$. We give an illustration in \cref{fig:gadget2}.\cqed
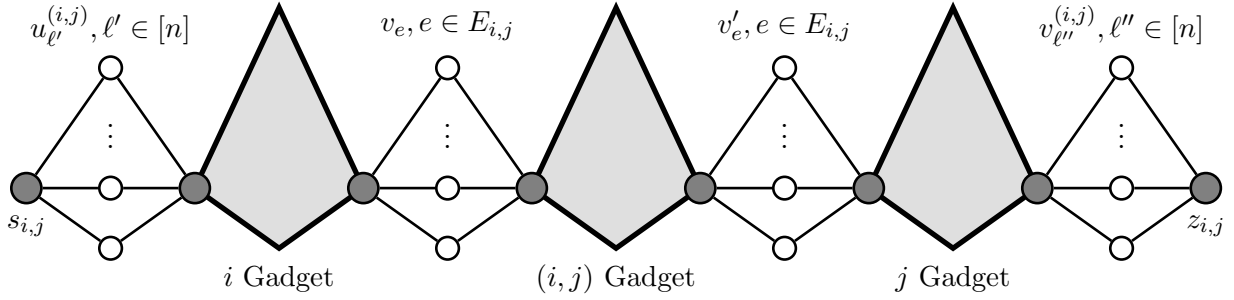
\begin{figure}[t]
\begin{center}
\begin{tikzpicture}[line width=1pt, scale=.8, xscale=1.4]
\draw[line width=2pt,fill=gray!25] (3,4) -- (4,7) -- (5,4) -- (4,3) -- (3,4);
\draw[line width=2pt,fill=gray!25] (-1,4) -- (0,7) -- (1,4) -- (0,3) -- (-1,4);
\draw[line width=2pt,fill=gray!25] (7,4) -- (8,7) -- (9,4) -- (8,3) -- (7,4);

    \node[vert2] (C1) at (3,4) {};
    \node[vert2] (C2) at (5,4) {};
    \node (L) at (4,2.5) {$(i,j)$ Gadget};

    \node[vert2] (C01) at (-1,4) {};
    \node[vert2] (C02) at (1,4) {};
    \node (L0) at (0,2.5) {$i$ Gadget};

    \node[vert2] (C11) at (7,4) {};
    \node[vert2] (C12) at (9,4) {};
    \node (L1) at (8,2.5) {$j$ Gadget};

    \node[vert] (E1) at (2,3) {};
    \node[vert] (E2) at (2,4) {};
    \node (D) at (2,5) {$\vdots$};
    \node[vert] (E3) at (2,6) {};
    \node (E4) at (2,6.7) {$v_e, e\in E_{i,j}$};

    \node[vert] (E21) at (6,3) {};
    \node[vert] (E22) at (6,4) {};
    \node (D2) at (6,5) {$\vdots$};
    \node[vert] (E23) at (6,6) {};
    \node (E24) at (6,6.7) {$v'_e, e\in E_{i,j}$};

    \node[vert] (S1) at (-2,3) {};
    \node[vert] (S2) at (-2,4) {};
    \node (D3) at (-2,5) {$\vdots$};
    \node[vert] (S3) at (-2,6) {};
    \node (S4) at (-2,6.7) {$u^{(i,j)}_{\ell'}, \ell'\in [n]$};
    \node[vert2,label=below:$s_{i,j}$] (S) at (-3,4) {}; 

    \node[vert] (T1) at (10,3) {};
    \node[vert] (T2) at (10,4) {};
    \node (D3) at (10,5) {$\vdots$};
    \node[vert] (T3) at (10,6) {};
    \node (T4) at (10,6.7) {$v^{(i,j)}_{\ell''}, \ell''\in [n]$};
    \node[vert2,label=below:$z_{i,j}$] (T) at (11,4) {}; 




\draw (S) -- (S1);
\draw (S) -- (S2);
\draw (S) -- (S3);
\draw (C01) -- (S1);
\draw (C01) -- (S2);
\draw (C01) -- (S3);

\draw (C02) -- (E1);
\draw (C02) -- (E2);
\draw (C02) -- (E3);
\draw (C1) -- (E1);
\draw (C1) -- (E2);
\draw (C1) -- (E3);

\draw (C2) -- (E21);
\draw (C2) -- (E22);
\draw (C2) -- (E23);
\draw (C11) -- (E21);
\draw (C11) -- (E22);
\draw (C11) -- (E23);

\draw (C12) -- (T1);
\draw (C12) -- (T2);
\draw (C12) -- (T3);
\draw (T) -- (T1);
\draw (T) -- (T2);
\draw (T) -- (T3);
\end{tikzpicture}
    \end{center}
    \caption{Illustration of vertices and edges added for color combination $i,j$. Notice that gray vertices form a vertex cover.}\label{fig:gadget2}
\end{figure}
\end{construction}
The temporal graph $\TG$ described in \cref{constr} can clearly be computed in polynomial time.
Note that $|S|\in O(k^2)$. Furthermore, let $X$ be the set of all vertices that are entry or exit to some gadget. Note that $X\cup S$ forms a vertex cover for the underlying graph of $\TG$ and we have that $|X\cup S|\in O(k^2)$. 

Before we prove the correctness of the reduction, we show the following two properties of the constructed temporal graph $\TG$. Intuitively, we show that once a temporal path enters a gadget, it has to traverse the body of the gadget and then exit the gadget.
\begin{lemma}\label{claim3}
Let $P$ be a temporal path in $\TG$ that visits the entry vertex $c$ of some gadget $H(p,q,t)$ present in $\TG$ and afterwards at least two more vertices. Then after $c$, the path $P$ visits a vertex $w^{(\ell)}_r$ for some $\ell\in\{0,1,\ldots,p\}$ and $r\in [q]$ of the body of the gadget, and then the exit vertex $c'$ of the gadget.
\end{lemma}
\begin{proof}
Let $P$ be a temporal path in $\TG$. 
We show the lemma assuming that $P$ visits the entry vertex $c$ of the gadget corresponding to color $i$ and afterwards at least two more vertices. 
By an analogous argument we can also show that the claim holds if $P$ visits the entry vertex of a gadget corresponding to a color combination $i,j$.

Assume for contradiction that $P$ does not visit a vertex of the body of the gadget. Then by construction it is either a vertex $v_\ell^{(i,j)}$ for some $j>i$, or some $v'_e$ for some $e\in E_{j,i}$ with $j<i$. Note that both $v_\ell^{(i,j)}$ and $v'_e$ have degree two in the underlying graph, hence there is only one possible way for a path to continue. 
In the first case where $P$ visits $v_\ell^{(i,j)}$ for some $j>i$, the path would need to continue to $s_{i,j}$. However, the edge $\{v_\ell^{(i,j)},s_{i,j}\}$ has a smaller label than the edge $\{c,v_\ell^{(i,j)}\}$. Hence we can rule out that $P_{s,z}$ visits $v_\ell^{(i,j)}$.
In the second case where $P$ visits $v'_e$ for some $e\in E_{j,i}$ with $j<i$, the path would need to continue to the exit of the gadget corresponding to the color combination $j,i$. Again, we have that the second edge that the path would have to use has a smaller label than the first one. Hence we can also rule out that $P$ visits $v'_e$. We can conclude that the vertex visited by $P$ directly after $c$ is some vertex of the body of the gadget. All vertices in the body of the gadget have degree two, hence we have that the next vertex visited by $P$ is the exit vertex $c'$.
\end{proof}

\begin{lemma}\label{claim4}
Let $P$ be a temporal path in $\TG$ that visits the exit vertex $c'$ of some gadget $H(p,q,t)$ present in $\TG$ and before at least two more vertices. Then before $c'$, the path $P$ visited a vertex $w^{(\ell)}_r$ for some $\ell\in\{0,1,\ldots,p\}$ and $r\in [q]$ of the body of the gadget, and before that it visited the entry vertex $c$ of the gadget.
\end{lemma}
\begin{proof}
    This lemma can be proven in an analogous way as \cref{claim3}.
\end{proof}

Assuming $(\TGcompact, S)$ is a yes-instance of \TempDisjointPaths, we use \cref{claim3,claim4} to show an important property of temporal paths in a solution $\mathcal{P}$ to $(\TGcompact, S)$ that connect source-sink pairs $(s_{i,j},z_{i,j})$ with $1\le i<j\le k$. Intuitively, here we show that these temporal path traverse the gadgets as illustrated in \cref{fig:gadget2}.

\begin{lemma}\label{claim5}
Assume $(\TGcompact, S)$ is a yes-instance of \TempDisjointPaths\ and let $\mathcal{P}$ be a solution. 
Let $1\le i<j\le k$ and let $P\in \mathcal{P}$ be a temporal path in $\TG$ from $s_{i,j}$ to $z_{i,j}$. Then $P$ visits a vertex in the body of each of the following gadgets.
\begin{itemize}
    \item The gadget corresponding to color $i$.
    \item The gadget corresponding to color $j$.
    \item The gadget corresponding to color combination $i,j$. 
\end{itemize}
\end{lemma}
\begin{proof}
By construction of $\TG$, the temporal path $P$ first visits some vertex $v^{(i,j)}_{\ell'}$ and then visits the entry vertex $c_i$ of the gadget corresponding to color $i$ (see \cref{fig:gadget2}). Since $z_{i,j}$ is not directly connected to $c_i$, we have that $P$ visits at least two more vertices and by \cref{claim3} we have that it visits some vertex in the body of the gadget corresponding to color $i$ and then the exit vertex $c'_i$ of the gadget.

Similarly, path $P$ visits some vertex $v^{(i,j)}_{\ell''}$, before arriving at vertex $z_{i,j}$. Before visiting $v^{(i,j)}_{\ell''}$, path $P$ visited the exit vertex $c'_j$ for the gadget corresponding to color $j$. Since $s_{i,j}$ is not directly connected to $c'_j$, we have that $P$ visited at least two more vertices before and by \cref{claim4} we have that it visited some vertex in the body of the gadget corresponding to color $j$ and before that it visited the entry vertex $c_j$ of the gadget.

The exit vertex $c'_i$ of the gadget corresponding to color $i$ is connected to vertices $v_e$ for some $e\in E_{i,i'}$ with $i<i'$ and to vertices $v_{\ell}^{(i',i)}$ for some $i'<i$. We can exclude that $P$ visits a vertex of the latter kind, since those are afterwards only connected to vertices $z_{i',i}$ from which (by construction of $\TG$) $P$ would not be able to continue.
Hence, we can assume that $P$ visits a vertex $v_e$ for some $e\in E_{i,i'}$ with $i<i'$ after the exit vertex $c'_i$ of the gadget corresponding to color $i$. From there, path $P$ has to continue to the entry vertex $c_{i,j}$ of the gadget corresponding to color combination $i,i'$ and since $z_{i,i'}$ is not connected to $c_{i,i'}$, we have that $P$ visits at least two more vertices and by \cref{claim3} we have that it visits some vertex in the body of the gadget corresponding to color combination $i,i'$ and then the exit vertex of the gadget.

Now consider the gadget corresponding to color combination $i,i'$. By \cref{claim2} we have that at most two paths in $\mathcal{P}$ can traverse the gadget, one of which is the one corresponding to the source-sink pair of the gadget. Above we have shown that every path $P\in \mathcal{P}$ from $s_{i,j}$ to $z_{i,j}$ with $1\le i<j\le k$ visits at least one gadget corresponding to a color combination. If path $P$ visits two or more different color combination gadgets, then there is one color combination gadget that is visited in total by at least three paths. By \cref{claim2} this is a contradiction to all paths in $\mathcal{P}$ being pairwise temporally non-intersecting.
\end{proof}

Now we are ready to prove \cref{thm:w1hardnessVCN}.
\begin{proof}[Proof of \cref{thm:w1hardnessVCN}]
Let $(\TGcompact, S)$ be the \TempDisjointPaths\ instance described by \cref{constr} for \textsc{Multicolored Clique} instance $G=(V_1\uplus\ldots\uplus V_k, E)$. 
Recall that $(\TGcompact, S)$ can be computed in polynomial time and the vertex cover number of the underlying graph of $\TG$ is in $O(k^2)$. 
We show that $(\TGcompact, S)$ is a yes-instance of \TempDisjointPaths\ if and only if $G$ contains a clique of size $k$.

$(\Rightarrow)$: Assume $(\TGcompact, S)$ is a yes-instance of \TempDisjointPaths\ and let $\mathcal{P}$ be a solution. 
Note that in the constructed instance $(\TG, S)$ the multiset $S$ does not contain multiple copies of the same source-sink pair, that is, it is actually a set. Hence, for every $(s,z)\in S$ we use $P_{s,z}\in \mathcal{P}$ to denote the temporal path from $s$ to $z$ in $\mathcal{P}$. 

Consider the gadget for color $i$. Let $(s_i,z_i)$ be the source-sink pair of the gadget. 
Since $s_i$ is only connected to the entry vertex $c_i$ of the gadget and $z_i$ is not connected to $c_i$, we have that $P_{s_i,z_i}$ visits the entry vertex $c_i$ of the gadget and then at least two more vertices. By \cref{claim3} we have that $P_{s_i,z_i}$ visits a vertex of the body of the gadget and then the exit vertex $c'_i$. Afterwards $P_{s_i,z_i}$ needs to continue to $z_i$, otherwise the path would have to revisit $c'_i$.

We can conclude that $P_{s_i,z_i}$ has the following form: starting at $s_i$ it first visits the entry $c_i$ of the gadget, then it visits exactly one vertex of the body of the gadget, then it visits the exit $c'_i$ of the gadget, and finally it reaches $z_i$.

Let $w^{(\ell)}_r$ with $\ell\in[n]$ and $r\in\{0,1,\ldots,k-1\}$ be the vertex in the body of the gadget that is visited by $P_{s_i,z_i}$. We say that $w^{(\ell)}_r$ is \emph{selected} in the gadget corresponding to color $i$.

Let $X=\{v_\ell\in V_i\mid i\in[k] \wedge w^{(\ell)}_r \text{ is selected in the gadget corresponding to color } i \text{ for some } r\in\{0,1,\ldots,k-1\}\}$. 
We show that $X$ is a clique in $G$. 
Assume for contradiction that $X$ is not a clique in $G$. Then there are vertices $v_\ell\in X\cap V_i$ and $v_{\ell'}\in X\cap V_j$ for some $i,j\in [k]$ with $i\neq j$ such that $\{v_{\ell'},v_{\ell''}\}\notin E_{i,j}$.

Consider the source-sink pair $(s_{i,j},z_{i,j})$ and let $P_{i,j}$ denote the temporal path in $\mathcal{P}$ from $s_{i,j}$ to~$z_{i,j}$. By \cref{claim5} we know that $P_{i,j}$ first traverses the gadget corresponding to color $i$, then the gadget corresponding to color combination $i,j$, and lastly the gadget corresponding to color $j$.


By construction of $\TG$, there exists an $e\in E_{i,j}$ such that $P_{i,j}$ visits $v_e$ right before entering the gadget corresponding to color combination $i,j$ and $P_{i,j}$ visits $v_e'$ right after exiting the gadget corresponding to color combination $i,j$. In any other case, by \cref{claim2} we have that $P_{i,j}$ would temporally intersect with the temporal path in $\mathcal{P}$ corresponding to the source-sink pair of the gadget.

Now by the construction of $\TG$ and \cref{claim2} we have that $v_{\ell'}$ and $v_{\ell''}$ are the two endpoints of edge $e$, otherwise $P_{i,j}$ would temporally intersect with the temporal path in $\mathcal{P}$ corresponding to the source-sink pair of one or both of the gadgets corresponding to colors $i$ and $j$. Hence, we have that $\{v_{\ell'},v_{\ell''}\}\in E_{i,j}$, a contradiction.

$(\Leftarrow)$: Assume $G$ is a yes-instance of \textsc{Multicolored Clique} and let $X\subseteq\bigcup_{i=1}^k V_i$ with $|X\cap V_i|=1$ for all $i\in [k]$ be a multicolored clique in $G$. 

We use the following temporal paths to connect source-sink pairs of gadgets corresponding to colors. Let $i\in[k]$ and let $(s_i,z_i)$ be the source-sink pair of the gadget corresponding to color $i$. Let $\{v_\ell^{(i)}\}=X\cap V_i$. Let $c_i$ denote the entry vertex and $c_i'$ the exit vertex of the gadget corresponding to color $i$. We use the path starting at $s_i$, then continuing to $c_i$, then to $w_0^{(\ell)}$, then to $c_i'$, and finally arriving at $z_i$. Note that since every edge in $\TG$ has exactly one time label, this uniquely specifies the temporal path. By construction of the gadgets, the time edges used by the path have strictly increasing time labels.
It is easy to see that temporal paths connecting the source-sink pairs of gadgets corresponding to different colors do not temporally intersect, since they do not visit common vertices.

We use the following temporal paths to connect source-sink pairs of gadgets corresponding to color combinations. Let $1\le i<j\le k$ and let $(s'_{i,j},z'_{i,j})$ be the source-sink pair of the gadget corresponding to color combination $i,j$. Let $\{v_{\ell'}^{(i)}\}=X\cap V_i$, let $\{v_{\ell''}^{(j)}\}=X\cap V_j$, and let $e_\ell^{(i,j)}=\{v_{\ell'}^{(i)},v_{\ell''}^{(j)}\}\in E_{i,j}$.
Let $c_{i,j}$ denote the entry vertex and $c_{i,j}'$ the exit vertex of the gadget corresponding to color combination $i,j$. We use the path starting at $s'_{i,j}$, then continuing to $c_{i,j}$, then to $w_0^{(\ell)}$, then to $c_{i,j}'$, and finally arriving at $z'_{i,j}$. Note that since every edge in $\TG$ has exactly one time label, this uniquely specifies the temporal path. By construction of the gadgets, the time edges used by the path have strictly increasing time labels.
It is easy to see that temporal paths connecting the source-sink pairs of gadgets corresponding to different colors combinations do not temporally intersect, since they do not visit common vertices.

We can further observe that temporal paths connecting source-sink pairs of gadgets corresponding to colors do not temporally intersect temporal paths connecting source-sink pairs of gadgets corresponding to color combinations, since they do not visit common vertices.

Next, we describe how to connect source-sink pairs $(s_{i,j},z_{i,j})$ with $1\le i<j\le k$ with a temporal path $P_{i,j}$. Again, let $\{v_{\ell'}^{(i)}\}=X\cap V_i$, let $\{v_{\ell''}^{(j)}\}=X\cap V_j$, and let $e_\ell^{(i,j)}=\{v_{\ell'}^{(i)},v_{\ell''}^{(j)}\}\in E_{i,j}$.
See \cref{fig:gadget2} for an illustration of the gadgets and vertices traversed and visited by $P_{i,j}$. Starting at $s_{i,j}$, the path $P_{i,j}$ continues to $v_{\ell'}^{(i,j)}$, then it traverses the gadget corresponding to color $i$ in the following way: it continues to the entry vertex of the gadget, then to the vertex $w_{j-1}^{(\ell')}$ in the body of the gadget, and then to the exit vertex of the gadget. After the gadget corresponding to color $i$, the path continues to vertex $v_e$. Then it traverses the gadget corresponding to color combination $i,j$ in the following way: it continues to the entry vertex of the gadget, then to the vertex $w_1^{(\ell)}$ in the body of the gadget, and then to the exit vertex of the gadget. After the gadget corresponding to color combination $i,j$, the path continues to vertex $v_e'$. Then it traverses the gadget corresponding to color $j$ in the following way: it continues to the entry vertex of the gadget, then to the vertex $w_i^{(\ell'')}$ in the body of the gadget, and then to the exit vertex of the gadget. Finally, the path continues to $v_{\ell''}^{(i,j)}$ and then to $z_{i,j}$.

First, we argue that $P_{i,j}$ is indeed a temporal path, that is, the labels on the time edges used by the path are strictly increasing. The path traverses the following time edges:
\begin{enumerate}
    \item $\{s_{i,j},v_{\ell'}^{(i,j)}\}$ with label one.
    \item $\{v_{\ell'}^{(i,j)}, c_i\}$, where $c_i$ denotes the entry vertex of the gadget corresponding to color $i$, with label $2(i-1)(kn+m)+2k+(2\ell'-1)(k-1)+2(j-1)-2$.
    \item $\{c_i, w_{j-1}^{(\ell')}\}$ in the body of the gadget corresponding to color $i$, with label $2(i-1)(kn+m)+2k+(2\ell'-1) (k-1)+2(j-1)-1$.
    \item $\{w_j^{(\ell')},c'_i\}$ in the body of the gadget corresponding to color $i$, with label $2(i-1)(kn+m)+2k+(2\ell'-1) (k-1)+2(j-1)$.
    \item $\{c_i',v_e\}$, where $c'_i$ denotes the exit vertex of the gadget corresponding to color $i$, with label $2(i-1)(kn+m)+2k+(2\ell'-1)(k-1)+2(j-1)+1$.
    \item $\{v_e, c_{i,j}\}$, where $c_{i,j}$ denotes the entry vertex of the gadget corresponding to color combination $i,j$, with label $2(i-1)(kn+m)+2kn+2\ell-1$.
    \item $\{c_{i,j},w_1^{(\ell)}\}$ in the body of the gadget corresponding to color combination $i,j$, with label $2(i-1)(kn+m)+2kn+2\ell$.
    \item $\{w_1^{(\ell)},c_{i,j}\}$ in the body of the gadget corresponding to color combination $i,j$, with label $2(i-1)(kn+m)+2kn+2\ell+1$.
    \item $\{c_{i,j}',v_e'\}$, where $c'_{i,j}$ denotes the exit vertex of the gadget corresponding to color combination $i,j$, with label $2(i-1)(kn+m)+2kn+2\ell+2$.
    \item $\{v_e', c_j\}$, where $c_j$ denotes the entry vertex of the gadget corresponding to color $j$, with label $2(j-1)(kn+m)+2k+(2\ell''-1)(k-1)+2i-2$.
    \item $\{c_j, w_i^{(\ell'')}\}$ in the body of the gadget corresponding to color $j$, with label $2(j-1)(kn+m)+2k+(2\ell''-1) (k-1)+2i-1$.
    \item $\{w_i^{(\ell'')},c'_j\}$ in the body of the gadget corresponding to color $j$, with label $2(j-1)(kn+m)+2k+(2\ell''-1) (k-1)+2i$.
    \item $\{c_j',v_{\ell''}^{(i,j)}\}$, where $c'_j$ denotes the exit vertex of the gadget corresponding to color $j$, with label $2(j-1)(kn+m)+2k+(2\ell''-1)(k-1)+2i+1$.
    \item $\{v_{\ell''}^{(i,j)},z_{i,j}\}$ with label $4n^3$.
\end{enumerate}
In most cases it is obvious that the labels are strictly increasing. For the 5th and 6th label, note that we assume $k<n$. For the 9th and 10th label, note that we have that $i<j$.

Next, we show that $P_{i,j}$ does not temporally intersect other temporal paths we constructed. Note that $P_{i,j}$ only shares vertices with temporal paths that also traverse the gadget corresponding to color $i$, the gadget corresponding to color $j$, and temporal paths that also traverse the gadget corresponding to color combination $i,j$.
Temporal paths that traverse the gadget corresponding to color $i$ are temporal paths $P_{i,j'}$ for some $i<j'\neq j$, temporal paths $P_{i',i}$ for some $i'<i$, and the temporal path connecting the source-sink pair of the gadget corresponding to colors $i$. 
Similarly, temporal paths that traverse the gadget corresponding to color $i$ are temporal paths $P_{i',j}$ for some $i\neq i'<j$, temporal paths $P_{j,j'}$ for some $j<j'$, and the temporal path connecting the source-sink pair of the gadget corresponding to colors $j$.
Lastly, the only other temporal path traversing the gadget corresponding to color combination $i,j$ is the one connecting the source-sink pair of the gadget.

Consider the gadget corresponding to color $i$. The only vertices shared by the constructed temporal paths are the entry vertex $c_i$ and the exit vertex $c_i'$ of the gadget.
The temporal path $P_i$ connecting the source-sink pair of the gadget visits $c_i$ from time one to time $2(i-1)(kn+m)+2k+(2\ell'-1) (k-1)-1$. Temporal paths $P_{i,j'}$ with $i<j'$ visit $c_i$ from time $2(i-1)(kn+m)+2k+(2\ell'-1) (k-1)+2(j'-1)-2$ to time $2(i-1)(kn+m)+2k+(2\ell'-1) (k-1)+2(j'-1)-1$. Hence, they do not temporally intersect $P_i$ (since $j'>1$) in $c_i$ and they do not temporally intersect each other in $c_i$.
Similarly, temporal paths $P_{i',i}$ with $i'<i$ visit $c_i$ from time $2(i-1)(kn+m)+2k+(2\ell'-1) (k-1)+2i'-2$ to time $2(i-1)(kn+m)+2k+(2\ell'-1) (k-1)+2i'-1$. Hence, they also do not temporally intersect $P_i$ in $c_i$ and they do not temporally intersect each other in $c_i$. Furthermore, temporal paths of the form $P_{i,j'}$ do not temporally intersect temporal paths of the form $P_{i',i}$ in $c_i$, since $j'>i$ and $i'<i$. By inspecting the arrival and departure times of the constructed temporal paths at exit vertex $c'_i$ of the gadget, we can conclude in an analogous way that the paths do not temporally intersect in $c_i'$.

By analogous arguments, we can show that none of the constructed temporal paths temporally intersect $P_{i,j}$ in the gadget corresponding to color $j$.

Lastly, consider the gadget corresponding to color combination $i,j$. The only temporal path that $P_{i,j}$ shares vertices with of that gadget is the temporal path $P$ connecting the source-sink pair of the gadget. Again, the two temporal paths both visit the entry vertex $c$ and the exit vertex $c'$ of the gadget. Path $P$ visits $c$ from time 1 to time $2(i-1)(kn+m)+2kn+2\ell-2$. The temporal path $P_{i,j}$ visits $c$ from time $2(i-1)(kn+m)+2kn+2\ell-1$ to time $2(i-1)(kn+m)+2kn+2\ell$. Hence, the two temporal paths do not temporally intersect in $c$. By inspecting the arrival and departure times of the constructed temporal paths at exit vertex $c'$ of the gadget, we can conclude in an analogous way that the paths do not temporally intersect in $c'$.

It follows that the constructed temporal paths pairwise do not temporally intersect. This concludes the proof.
\end{proof}

\subsection{Parameterized Hardness of \TempDisjointWalks}\label{sec:hardnesswalks}
\citet{KlobasMMNZ23} left the parameterized complexity of \tdw with respect to the number $|S|$ of source-sink pairs on temporal trees as an open question.
In this section, we answer this question by showing that the problem is \W{1}-hard for this parameterization, even on temporal stars.
This may be somewhat surprising considering that \citet{KlobasMMNZ23} showed that \tdp is \fpt on trees.
This implies that while \tdp is harder than \tdw on arbitrary graphs for $|S|$ as the parameter (the former is \NP-hard for $|S|=2$, while the latter is solvable in polynomial time for constant $|S|$), \tdw is harder than \tdp on temporal trees.

\begin{theorem}
	\label{thm:k-star}
	\tdw{} on temporal stars is \W{1}-hard when parameterized by the number $|S|$ of source-sink pairs.
\end{theorem}

We will give a parameterized reduction from the \W{1}-hard~\cite{fellows2009multipleinterval} problem \textsc{Multicolored Clique}.
The input for this problem consists of an integer $k$ and a properly $k$-colored graph $G=(V_1\uplus V_2 \uplus \ldots \uplus V_k,E)$ and one is asked to decide whether $G$ contains a clique of size~$k$.
Any such clique must, of course, contain exactly one vertex from each color class.

\begin{construction}
	\label{constr:star}
	Let $(G=(V_1\uplus V_2 \uplus \ldots \uplus V_k,E),k)$ be an instance of \textsc{Multicolored Clique}.
	We may assume w.l.o.g.\ 
 that $\abs{V_1} = \abs{V_2} = \ldots \abs{V_k} \eqqcolon n$.
	Suppose that $V_i = \{v^i_1,\ldots,v^i_n\}$.
	We will now construct an instance $(\TGcompact,S)$ of \tdw.
	
	We start by describing $S$.
	For every $i \in [k]$, there are two terminal pairs $(s_i,z_i)$ and $(\ts_i,\ttt_i)$.
	The temporal walks that connect these two pairs will encode the selection of a vertex in $V_i$.
	Additionally, for every $i,j \in [k]$ with $i<j$ there is a terminal pair $(s_{i,j},z_{i,j})$.
	The temporal walk for this pair verifies that at least one vertex has been selected in each of $V_i$ and $V_j$ and that those two vertices are adjacent.
	Let $S$ denote this set of terminal pairs.
	
	Next we will define $\calG=(V,E_1,\ldots,E_T)$.
	We start by giving $V$.
	For every $i\in [k]$ there are two sets of vertices, the first, $W_i$, is intended to be used by the temporal walk connecting $(s_i,z_i)$ and the other, $\tW_i$, by the temporal walk for $(\ts_i,\ttt_i)$.
	Let $W_i \coloneqq \{w^i_1,\ldots,w^i_{kn},x^i_1,\ldots,x^i_{kn},y^i_1,\ldots,y^i_n\}$ and $\tW_i \coloneqq \{\tw^i_1,\ldots,\tw^i_{kn},\tx^i_1,\ldots,\tx^i_{kn},\ty^i_1,\ldots,\ty^i_n\}$.
	Then, for every edge $e=\{u,v\} \in E$ with $u\in V_i$ and $v\in V_j$ for some $i<j$ there are vertices $W_e \coloneqq \{\alpha_u^j,\beta_e,\gamma_v^i\}$.
	Finally, there is a central vertex $c$, to which every edge will be incident.
	Let $V \coloneqq \{c\} \cup \bigl(\bigcup_{(s,z) \in S} \{s,z\}\bigr) \cup \bigl(\bigcup_{i \in [k]} W_i \cup \tW_i\bigr) \cup \bigl(\bigcup_{e\in E} W_e\bigr)$.
	
	It remains to define $E_1,\ldots,E_T$.
	For every $i \in [k]$, there is a sequence of edge sets $E^i_0,\ldots,E^i_{4kn+6}$.
	Informally speaking, in this sequence the temporal walks connecting $(s_i,z_i)$ and $(\ts_i,\ttt_i)$ select a vertex in $V_i$. The temporal walk connecting $(s_{i,j},z_{i,j})$ for $j \neq i$ verifies that the selected vertex is adjacent to the one selected in $V_j$.
	First, for $\ell \in [kn]$, the vertices $w_\ell^i$ are adjacent to $c$ in the layers $E_{4\ell-3}$ and $E^i_{4\ell}$, and the vertices $\tw_\ell^i$ have an edge to $c$ in $E^i_{4\ell-1}$ and $E^i_{4\ell+2}$.
	Next, for $\ell \in [kn]$, the vertices $x^i_\ell$ are adjacent to $c$ in $E^i_{4\ell-2}$ and $E^i_{4\ell+1}$ and $\tx_{\ell}^i$ have edges to $c$ in $E^i_{4\ell}$ and $E^i_{4\ell+3}$.
	Finally, for $\ell \in [n]$, the vertices $y^i_\ell$ are adjacent to $c$ in the layers $E^i_{4k(\ell-1)+1}$ and $E^i_{4k\ell+1}$, while for $\ty^i_\ell$ those layers are $E^i_{4k(\ell-1)+3}$ and $E^i_{4k\ell+3}$.
	Additionally, the starting vertex $s_i$ is adjacent to $c$ in the layer $E^i_0$ and $z_i$ has an edge to $c$ in $E^i_{4kn+4}$.
	Similarly, for $\ts_i$ and $\ttt_i$, those layers are $E^i_2$ and $E^i_{4kn+6}$, respectively.

	For any $i,j \in [k]$ with $i<j$, there is a layer $E^{i,j}_1$, which contains the edge $\{s_{i,j},c\}$ and a second subsequent layer $E^{i,j}_2$, which contains  $\{c,\alpha_{v}^j\}$ for all $v\in V_i$ that have a neighbor in $V_j$.
	There is also a layer $E_{f-1}^{i,j}$, which connects  $c$ to $\gamma_v^i$ for all $v\in V_j$ that have a neighbor in $V_i$, and finally $E_f^{i,j}$ connecting $z_{i,j}$ to $c$.
	Next, consider edge $e \in E$.
	Suppose that one endpoint of that edge is $v^i_a \in V_i$ and the other endpoint is $v^j_b \in V_j$, with $i < j$.
	Then, there is an edge from $\alpha_{v^i_a}^j$ to $c$ in the layer $E^i_{4k(a-1)+4(j-1)}$, from $c$ to $\beta_e$ in layer $E^i_{4k(a-1)+4(j-1)+2}$, from $\beta_e$ to $c$ in $E^j_{4k(b-1)+4i}$, and from $c$ to $\gamma_{v^j_b}^i$ in $E^j_{4k(b-1)+4i+2}$.
	
	The order of the layers in $\calG$ is as follows.
	The layers $E^{i,j}_1$ and $E^{i,j}_2$ for each $i,j\in k$ are consecutive to one another and all such layers come at the very beginning of the temporal graph.
	Then, come the layers $E^1_1,\ldots,E^1_{4kn+6}$, followed by $E^2_1,\ldots,E^2_{4kn+6}$, and so on.
	The temporal graph concludes with the layers $E^{i,j}_{f-1}$ and $E^{i,j}_f$ for each $i,j\in k$ being consecutive to one another. We give an illustration of the construction in \cref{fig:vertexselection}.
	\cqed
\end{construction}

We will now give a brief overview of the intuition as to why this construction is correct before proving this claim formally.
First, consider for any $i \in [k]$ temporal walks $P_i$ and $\tP_i$ that connect $(s_i,z_i)$ and $(\ts_i,\ttt_i)$, respectively.
For the purpose of explanation, assume for now that the vertices $y^i_\ell$ and $\ty^i_\ell$ did not exist.
The first walk, $P_i$, must move from $s_i$ to $c$ in layer $E^i_0$, because $s_i$ is subsequently isolated.
The second walk, $\tP_i$, must similarly arrive in $c$ in $E^i_2$.
Hence, $P_i$ must leave $c$ in $E^i_1$, otherwise it temporally intersects $\tP_i$.
It must move to $w^i_1$. 
That vertex is again adjacent to $c$ in layer $E^i_4$ and isolated after that.
Hence, $P_i$ must return to $c$ in that layer.
Therefore, $\tP_i$ must leave $c$ in layer $E^i_3$, otherwise it temporally intersects $P_i$.
It must move to $\tw^i_1$. 
We can see that the two walks $P_i,\tP_i$ are ``locked'' into alternatingly moving from the center $c$ to vertices $w^i_\ell$ and $\tw^i_\ell$, respectively. For an illustration see \cref{fig:vertexselection}.

Now consider that the vertices $y^i_\ell$ and $\ty^i_\ell$ exist.
The pattern in which $P_i$ and $\tP_i$ move can be ``broken'' if $P_i$ or $\tP_i$ moves to $y^i_\ell$ or $\ty^i_\ell$, respectively. Now we can make two observations.
\begin{itemize}
    \item If the two walks do so almost simultaneously it creates a interval of size $O(k)$ where neither of the two walks occupy the center vertex $c$. This interval corresponds to a vertex in $V_i$. Hence, by choosing when to move to $y^i_\ell$ and $\ty^i_\ell$, respectively, a vertex from $V_i$ is ``selected''.
    \item After the two walks move back to $c$, they are locked in a similar pattern, where they alternatingly move from the center to vertices $x^i_\ell$ and $\ty^i_\ell$, respectively. (This happens also if only one of the walks move to $y^i_\ell$ or $\ty^i_\ell$.) From this pattern, they cannot move to vertices $y^i_\ell$ or $\ty^i_\ell$, hence at most one vertex is selected per color.
\end{itemize}
If no vertex is selected, that is, no vertices $y^i_\ell$ or $\ty^i_\ell$ are visited, then the temporal walk $P_{i,j}$ from $s_{i,j}$ to $z_{i,j}$ for any $i<j$ will temporally intersect $P_i$ or $\tP_i$. (If $i=k$ we make an analogous observation, where $i$ and $j$ exchange their role.)
Hence, we can assume that one vertex of every color is selected.
We can observe that $P_{i,j}$ first must move from $s_{i,j}$ to $c$ and then to $\alpha_e$ where $e$ is any edge connecting a vertex in $V_i$ to a vertex in $V_j$. Informally speaking, this is only possible without temporally intersecting any of the temporal walks $P_i,\tP_i,P_j,\tP_j$ if the endpoints of $e$ are selected as vertices for colors $i$ and $j$, respectively. Thereby, we verify that the selected vertices indeed form a clique in $G$.

\begin{figure}[t]
    \centering
    \begin{tikzpicture}[line width=1pt, xscale=.4]
    \node (Y1) at (-1.5,5) {$y^2_\ell$};
    \node (X1) at (-1.5,4) {$x^2_\ell$};
    \node (W1) at (-1.5,3) {$w^2_\ell$};
    \node (C) at (-1.5,2) {$c$};
    \node (Y2) at (-1.5,-1) {$\ty^2_\ell$};
    \node (X2) at (-1.5,0) {$\tx^2_\ell$};
    \node (W2) at (-1.5,1) {$\tw^2_\ell$};
		\path
			(0,2) node[tvert](b0) {}
			(2,2) node[tvert](b1) {}
			(4,2) node[tvert](b2) {}
			(6,2) node[tvert](b3) {}
			(8,2) node[tvert](b4) {}
			(10,2) node[tvert](b5) {}
			(12,2) node[tvert](b6) {}
			(14,2) node[tvert](b7) {}
			(16,2) node[tvert](b8) {}
			(18,2) node[tvert](b9) {}
			(20,2) node[tvert](b10) {}
			(22,2) node[tvert](b11) {}
			(24,2) node[tvert](b12) {}
			(26,2) node[tvert](b13) {}
			(28,2) node[tvert](b14) {}
			(30,2) node[tvert](b15) {}
			(32,2) node[tvert](b16) {}
			(34,2) node[tvert](bn) {};

        \draw[lightgray,dashed,line width=2pt] (8,6) -- (9,2);
        \node (A) at (8,6.25) {$\beta_{e'}$};
        \draw[lightgray,dashed,line width=2pt] (6,6) -- (9,2);
        \node (A) at (6,6.25) {$\beta_{e}$};
        \draw[lightgray,dashed,line width=2pt] (13,2) -- (14,6);
        \node (B1) at (14,6.25) {$\gamma_v^1$};
  
        \draw[gray,dotted,line width=2pt] (16,6) -- (17,2);
        \node (B2) at (16,6.25) {$\alpha_v^3$};
        \draw[gray,dotted,line width=2pt] (21,2) -- (22,6);
        \node (G) at (22,6.25) {$\beta_{e''}$};
        \draw[gray,dotted,line width=2pt] (21,2) -- (24,6);
        \node (G) at (24,6.25) {$\beta_{e'''}$};
        
		\draw[blue] (0,3) -- (1,2);

		\draw[blue] (3,2) -- (4,3);  
		\draw[blue] (4,3) -- (8,3);
		\draw[blue] (8,3) -- (9,2);
  
		\draw[blue] (11,2) -- (12,3);  
		\draw[blue] (12,3) -- (16,3);
		\draw[blue] (16,3) -- (17,2);
  
		\draw[blue] (19,2) -- (20,3); 
		\draw[blue] (20,3) -- (24,3); 
		\draw[blue] (24,3) -- (25,2); 
  
		\draw[blue] (27,2) -- (28,3); 
		\draw[blue] (28,3) -- (32,3); 
		\draw[blue] (32,3) -- (33,2); 
  

    	\draw[cyan] (2,4) -- (0,4);
    	\draw[cyan] (2,4) -- (3,2);

		\draw[cyan] (5,2) -- (6,4);  
		\draw[cyan] (6,4) -- (10,4);
		\draw[cyan] (10,4) -- (11,2);
  
		\draw[cyan] (13,2) -- (14,4);  
		\draw[cyan] (14,4) -- (18,4);
		\draw[cyan] (18,4) -- (19,2);
  
		\draw[cyan] (21,2) -- (22,4); 
		\draw[cyan] (22,4) -- (26,4); 
		\draw[cyan] (26,4) -- (27,2); 

    	\draw[cyan] (29,2) -- (30,4); 
		\draw[cyan] (30,4) -- (34,4); 
  
  		\draw[red] (4,1) -- (0,1);
  		\draw[red] (4,1) -- (5,2);

		\draw[red] (7,2) -- (8,1);  
		\draw[red] (8,1) -- (12,1);
		\draw[red] (12,1) -- (13,2);
  
		\draw[red] (15,2) -- (16,1);  
		\draw[red] (16,1) -- (20,1);
		\draw[red] (20,1) -- (21,2);
  
		\draw[red] (23,2) -- (24,1);
		\draw[red] (24,1) -- (28,1);
		\draw[red] (28,1) -- (29,2);
  
		\draw[red] (31,2) -- (32,1);
		\draw[red] (32,1) -- (34,1);

    	\draw[orange] (1,2) -- (2,0);
    	\draw[orange] (2,0) -- (6,0);
    	\draw[orange] (6,0) -- (7,2);

		\draw[orange] (9,2) -- (10,0);  
		\draw[orange] (10,0) -- (14,0);
		\draw[orange] (14,0) -- (15,2);
  
		\draw[orange] (17,2) -- (18,0);  
		\draw[orange] (18,0) -- (22,0);
		\draw[orange] (22,0) -- (23,2);
  
		\draw[orange] (25,2) -- (26,0);
		\draw[orange] (26,0) -- (30,0);
		\draw[orange] (30,0) -- (31,2);
  
		\draw[orange] (33,2) -- (34,0);

		\draw[teal, line width=2pt] (0,5) -- (2,5); 
		\draw[teal, line width=2pt] (2,5) -- (3,2);  
  
		\draw[teal, line width=2pt] (3,2) -- (4,5);  
		\draw[teal, line width=2pt] (26,5) -- (4,5);  
		\draw[teal, line width=2pt] (26,5) -- (27,2); 
  
		\draw[teal, line width=2pt] (27,2) -- (28,5); 
		\draw[teal, line width=2pt] (28,5) -- (34,5); 
 
		\draw[purple, line width=2pt] (0,-1) -- (6,-1);  
		\draw[purple, line width=2pt] (6,-1) -- (7,2); 

		\draw[purple, line width=2pt] (7,2) -- (8,-1);  
		\draw[purple, line width=2pt] (30,-1) -- (8,-1);  
		\draw[purple, line width=2pt] (30,-1) -- (31,2); 
   
		\draw[purple, line width=2pt] (31,2) -- (32,-1);  
		\draw[purple, line width=2pt] (32,-1) -- (34,-1); 

  \draw[edge, line width=2pt] (b0) edge node [midway, below] {} (bn);

\end{tikzpicture}
    \caption{Illustration of the part of the temporal graph $\TG$ as defined by \cref{constr:star} that corresponds to color $i=2$ for $k=3$ colors. Vertices are represented by horizontal lines. The horizontal position of the lines indicate the ``vertex type'' ($w^i_\ell,x^i_\ell,y^i_\ell,\tw^i_\ell,\tx^i_\ell,\ty^i_\ell$ for $\ell\in[n]$, and $c$), as described to the left. The non-horizontal lines represent time edges, where the label corresponds to the position of their connection to the center vertex $c$ (black horizontal line in the middle). Positions further to the left correspond to earlier time labels. Edges $e,e'$ connect vertices of color~1 to vertex $v$ of color 2. Edges $e'',e'''$ connect vertices of color 3 to vertex $v$ of color 2.}
    \label{fig:vertexselection}
\end{figure}
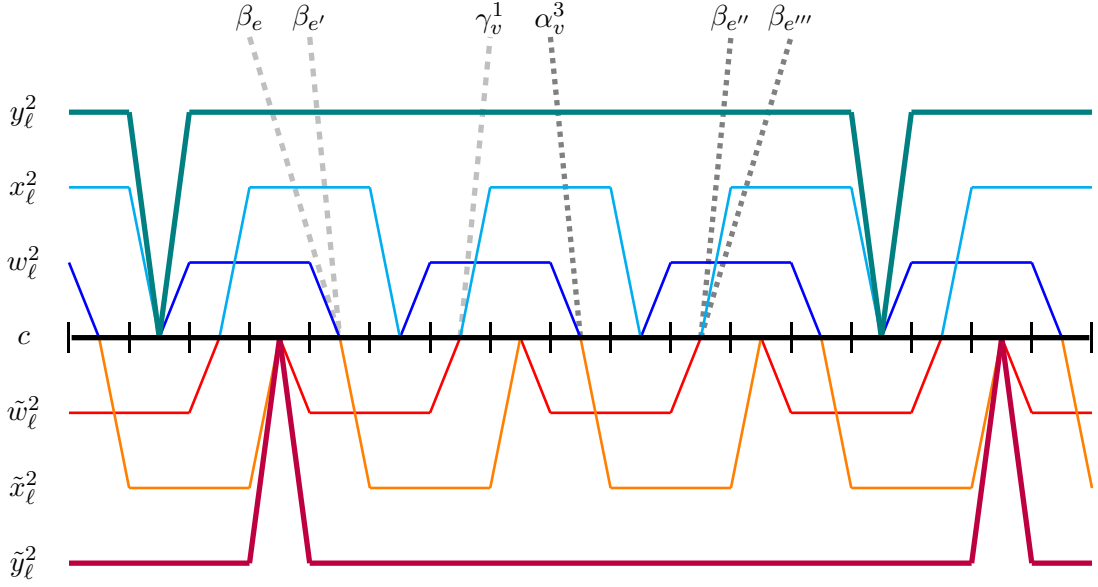

Now we prove the correctness of our reduction. We start by formally showing that the above claimed properties of temporal walks $P_i,\tP_i$ connecting $(s_i,z_i)$ and $(\ts_i,\ttt_i)$, respectively, hold.
\begin{lemma}\label{lemma1}
    Let $P_i,\tP_i$ be temporal walks connecting $(s_i,z_i)$ and $(\ts_i,\ttt_i)$, respectively, such that $P_i$ and $\tP_i$ are temporally disjoint. Then the following holds.
    \begin{itemize}
        \item If $P_i$ visits no vertex $y^i_\ell$ with $\ell\in[n]$ and $\tP_i$ visits no vertex $\ty^i_{\ell}$ with $\ell\in[n]$, then $P_i$ visits all vertices $w^i_1,\ldots, w^i_{kn}$ and $\tP_i$ visits all vertices $\tw^i_1,\ldots, \tw^i_{kn}$.
        \item $P_i$ visits at most one vertex $y^i_\ell$ with $\ell\in[n]$ and no vertex $\ty^i_\ell$ with $\ell\in[n]$. 
        \item  If $P_i$ visits $y^i_\ell$, then $P_i$ visits all vertices $w^i_1,\ldots, w^i_{k(\ell-1)}$ and all vertices $x^i_{k\ell+1},\ldots, x^i_{kn}$. Furthermore, $\tP_i$ visits all vertices $\tw^i_1,\ldots, \tw^i_{k(\ell-1)}$ and all vertices $\tx^i_{k\ell+1},\ldots, \tx^i_{kn}$.
        \item $\tP_i$ visits at most one vertex $\ty^i_\ell$ with $\ell\in[n]$ and no vertex $\ty^i_\ell$ with $\ell\in[n]$.
        \item If $\tP_i$ visits $\ty^i_{\ell'}$, then $\tP_i$ visits all vertices $\tw^i_1,\ldots, \tw^i_{k(\ell'-1)}$ and all vertices $\tx^i_{k\ell'+1},\ldots, \tx^i_{kn}$. Furthermore, $P_i$ visits all vertices $w^i_1,\ldots, w^i_{k(\ell'-1)}$ and all vertices $x^i_{k\ell'+1},\ldots, x^i_{kn}$.
        \item If $P_i$ visits $y^i_\ell$ and $\tP_i$ visits $\ty^i_{\ell'}$, then $\ell=\ell'$.
    \end{itemize}
\end{lemma}
\begin{proof}
Before we start, observe that each edge in the constructed graph $\TG$ has at most two labels. This means that if a temporal walk visits a leaf vertex (that is not its source or sink), then it is determined when the walk arrives at that vertex and when it leaves the vertex. Inspecting the labels on edges $\{c,\alpha_v^j\}$, $\{c,\beta_e\}$, and $\{c,\gamma_v^j\}$, we can observe that $P_i$ and $\tP_i$ do not traverse these edges.

Assume that $P_i$ visits no vertex $y^i_\ell$ with $\ell\in[n]$ and $\tP_i$ visits no vertex $\ty^i_{\ell}$ with $\ell\in[n]$.
By construction, $P_i$ must move from $s_i$ to $c$ in layer $E^i_0$, because $s_i$ is subsequently isolated.
Similarly, $\tP_i$ must move to $c$ in $E^i_2$.
Hence, $P_i$ must leave $c$ in $E^i_1$, otherwise it temporally intersects $\tP_i$.
It must move to $w^i_1$, since we assume that it does not move to $y^i_1$.
That vertex is again adjacent to $c$ in layer $E^i_4$ and isolated after that.
Hence, $P_i$ must return to $c$ in that layer.
Therefore, $\tP_i$ must leave $c$ in layer $E^i_3$, otherwise it temporally intersects $P_i$.
It must move to $\tw^i_1$, since we assume that it does not move to $y^i_1$.
By repeating this argument, we can see that the two walks $P_i$ and $\tP_i$ are ``locked'' into alternatingly moving from the center $c$ to vertices $w^i_\ell$ and $\tw^i_\ell$, respectively. For an illustration see \cref{fig:vertexselection}. It follows that $P_i$ visits all vertices $w^i_1,\ldots, w^i_{kn}$ and $\tP_i$ visits all vertices $\tw^i_1,\ldots, \tw^i_{kn}$.

The temporal walks $P_i$ and $\tP_i$ can only deviate from the above pattern if $P_i$ visits a vertex $y^i_\ell$ with $\ell\in[n]$ or $\tP_i$ visits a vertex $\ty^i_{\ell}$ with $\ell\in[n]$. Assume that $P_i$ is the first walk that visits a vertex $y^i_\ell$ with $\ell\in[n]$. The case where $\tP_i$ visits a vertex $\ty^i_{\ell}$ with $\ell\in[n]$ first is analogous. By the arguments above we immediately get that $P_i$ visits all vertices $w^i_1,\ldots, w^i_{k(\ell-1)}$ and $\tP_i$ visits all vertices $\tw^i_1,\ldots, \tw^i_{k(\ell-1)}$. Now we take a closer look on when $P_i$ moves back to $c$ from $y^i_\ell$ and which consequences that as on how $P_i$ and $\tP_i$ continue. Let $t$ denote the time step when $P_i$ moves back to $c$ from $y^i_\ell$, that is, let $t$ be the larger of the two time labels on edge $\{y^i_\ell,c\}$. By construction, $t$ is also the smaller time label on edges $\{y^i_{\ell+1},c\}$ and $\{w^i_{k\ell+1},c\}$, see \cref{fig:vertexselection} for an illustration. It follows that $P_i$ can neither visit $y^i_{\ell+1}$ nor $w^i_{k\ell+1}$.
Now we investigate the behavior of $\tP_i$. Clearly, $\tP_i$ cannot occupy the center vertex $c$ at time $t$, since then it would temporally intersect $P_i$. It follows that at time $t$, walk $\tP_i$ occupies $\tw^i_{k\ell}$, $\tx^i_{k\ell}$, or $\ty^i_{\ell}$, see \cref{fig:vertexselection} for an illustration. If $\tP_i$ occupies $\tw^i_{k\ell}$ at time $t$, then $\tP_i$ would move to $c$ at time $t+1$ and temporally intersect $P_i$. We can conclude that $\tP_i$ occupies $\tx^i_{k\ell}$ or $\ty^i_{\ell}$ at time $t$. In both cases $\tP_i$ moves back to $c$ at time $t+2$. It follows that $P_i$ cannot occupy $c$ at time $t+2$. The only way to realize this is for $P_i$ to move to $x^i_{k\ell+1}$ at time $t+1$.
Then $P_i$ moves back to $c$ at time $t+4$, which means that $\tP_i$ must leave $c$ at time $t+3$. The only way to realize this is for $P_i$ to move to $\tx^i_{k\ell+1}$ and then move back to $c$ at time $t+6$.
By repeating this argument, we can see that the two walks $P_i$ and $\tP_i$ are now ``locked'' into alternatingly moving from the center $c$ to vertices $x^i_{\ell'}$ and $\tx^i_{\ell'}$, respectively, for $\ell'\ge k\ell+1$. It follows that $P_i$ visits all vertices $x^i_{k\ell+1},\ldots, x^i_{kn}$ and $\tP_i$ visits all vertices $\tx^i_{k\ell+1},\ldots, \tx^i_{kn}$. Furthermore, we have that $P_i$ visits no $y^i_{\ell'}$ with $\ell'\neq\ell$ and no vertex $\ty^i_{\ell'}$ with $\ell'\in[n]$, and if $\tP_i$ visits a vertex $\ty^i_{\ell'}$ for some $\ell'\in[n]$, then it needs to be $\ty^i_{\ell}$. By an analogous chain of arguments, we can show an analogous behavior if $P_i$ and $\tP_i$ switch their roles.
\end{proof}

\begin{lemma}\label{lemma2}
Assume $(\TGcompact, S)$ is a yes-instance of \TempDisjointWalks\ and let $\mathcal{P}$ be a solution.
    Let $P_i,\tP_i$ be temporal walks connecting $(s_i,z_i)$ and $(\ts_i,\ttt_i)$, respectively, such that $P_i$ and $\tP_i$ are temporally disjoint.
    Let $P_{i',j'}\in\mathcal{P}$ with $i'=i$ or $j'=i$ be a temporal walk connecting $(s_{i',j'},z_{i',j'})$.
    If one of $P_i$ and $\tP_i$ visits no vertex $y^i_\ell$ with $\ell\in[n]$ and no vertex  $\ty^i_\ell$ with $\ell\in[n]$, then $P_{i',j'}$ temporally intersects $P_i$ or~$\tP_i$.
\end{lemma}
\begin{proof}
Consider the case where $i'=i$, the case where $j'=i$ can be shown in an analogous way. Assume that $P_i$ visits no vertex $y^i_\ell$ with $\ell\in[n]$ and no vertex  $\ty^i_\ell$ with $\ell\in[n]$. The case where $\tP_i$ visits no such vertex is analogous.

Temporal walk $P_{i,j'}$ (with $i<j'$) starts at $s_{i,j'}$ and then moves to the center vertex $c$. In the next time step, $P_{i,j'}$ has to move to some vertex $\alpha_v^{j'}$ with $v\in V_i$, otherwise $P_{i,j'}$ temporally intersects some other temporal walk $P_{i'',j''}$ connecting $(s_{i'',j''},z_{i'',j''})$ or the temporal walk $P_1$ connecting $(s_1,z_1)$.
By construction, $P_{i,j'}$ returns from $\alpha_v^{j'}$ to $c$ at some time step $t$ that is larger than the time label on $\{s_i,c\}$ and smaller than the time label on $\{c,z_i\}$. More precisely, by construction, there exist a vertex $w_{\ell^\star}^i$ for some $\ell^\star\in[kn]$ such that the larger of the two labels on $\{w_{\ell^\star}^i,c\}$ is also $t$. Now consider the case where $\tP_i$ also visits no vertex $y^i_\ell$ with $\ell\in[n]$ and no vertex $\ty^i_\ell$ with $\ell\in[n]$. Then by \cref{lemma1} we have that $P_i$ visits all vertices $w^i_1,\ldots, w^i_{kn}$ and hence $P_i$ and $P_{i,j'}$ temporally intersect.

It remains to consider the case where $\tP_i$ visits some vertex $y^i_\ell$ with $\ell\in[n]$ or some vertex  $\ty^i_\ell$ with $\ell\in[n]$.
By \cref{lemma1} we have that $\tP_i$ visits exactly one vertex $y^i_{\ell^{\star\star}}$ with ${\ell^{\star\star}}\in[n]$ and no vertex $\ty^i_\ell$ with $\ell\in[n]$. Furthermore, \cref{lemma1} gives us that in this case $P_i$ visits all vertices $w^i_1,\ldots, w^i_{k({\ell^{\star\star}}-1)}$ and all vertices $x^i_{k{\ell^{\star\star}}+1},\ldots, x^i_{kn}$.
We immediately get that if $\ell^\star\le k({\ell^{\star\star}}-1)$, then $P_i$ and $P_{i,j'}$ temporally intersect. By construction, we also have that $t$ is exactly by one smaller than the larger label on the edge $\{x_{\ell^\star}^i,c\}$. Hence, if $\ell^\star\ge k{\ell^{\star\star}}+1$, then $P_i$ and $P_{i,j'}$ temporally intersect.

We can conclude that $k({\ell^{\star\star}}-1) <\ell^\star< k{\ell^{\star\star}}+1$.
For this remaining case we need to consider all temporal walks $P_{i',j'}\in\mathcal{P}$ with $i'=i$ or $j'=i$ connecting $(s_{i',j'},z_{i',j'})$.
Above we analysed the case that $i'=i$ and $i<j'$. 
We observed that, by construction, every temporal walk $P_{i,j'}$ returns from some $\alpha_v^{j'}$ to $c$ at some time step $t_{j'}$ and there exist a vertex $w_{\ell_{j'}^\star}^i$ for some $\ell_{j'}^\star\in[kn]$ such that the larger of the two labels on $\{w_{\ell_{j'}^\star}^i,c\}$ is also $t_{j'}$. Hence, we have that $k({\ell^{\star\star}}-1) <\ell_{j'}^\star< k{\ell^{\star\star}}+1$ holds for all $j'>i$. Furthermore, we must have $\ell_{j'}^\star\neq \ell_{j''}^\star$ for all $i<j'<j''$, since otherwise temporal walks $P_{i,j'}$ and $P_{i,j''}$ temporally intersect.

Now consider the case that $i'<i$ and $i=j'$. 
Similarly, here we can observe that, by construction, every temporal walk $P_{i',i}$ has to move from $c$ to some $\gamma_v^{i}$ at some time step $t_{i'}$ and there exist a vertex $\tw_{\ell_{i'}^\star}^i$ for some $\ell_{i'}^\star\in[kn]$ such that the larger of the two labels on $\{\tw_{\ell_{i'}^\star}^i,c\}$ is also $t_{i'}$.
By a similar chain of arguments as made before, we can arrive at the conclusion that $k({\ell^{\star\star}}-1) <\ell_{i'}^\star< k{\ell^{\star\star}}+1$ must hold for all $i'<i$, and $\ell_{i'}^\star\neq \ell_{i''}^\star$ must hold for all $i'<i''<i$. Finally, by construction of $\TG$, we have that $\ell_{i'}^\star\neq \ell_{j'}^\star$ must hold for all $i'<i<j'$. 
Note that $k({\ell^{\star\star}}-1) - k{\ell^{\star\star}}+1 = k-1$.
Since there are $k-1$ different temporal walks $P_{i',j'}$ with $i'=i$ or $j'=i$, we get the following. Each temporal walk $P_{i,j'}$ with $i<j'$ must leave $c$ by moving to a next possible vertex $\beta_e$ for some $e\in E$. Similarly, each temporal walk $P_{i',i}$ with $i'<i'$ must have arrived at $c$ by moving from a latest previously possible vertex $\beta_e$ for some $e\in E$.
This means that each $P_{i,j'}$ and $P_{i',i}$ occupies $c$ for three consecutive time steps. Furthermore, we have that $P_i$ cannot move from $c$ to vertices $\alpha_v^{j'}$, $\beta_e$, or $\gamma_v^{i'}$ with out temporally intersecting with some $P_{i,j'}$ or $P_{i',i}$. However, if $P_i$ moves from $c$ to some vertex $w^i_\ell,x^i_\ell,y^i_\ell,\tw^i_\ell,\tx^i_\ell,\ty^i_\ell$ for some $\ell\in[n]$, it has to move back to $c$ three time steps later. Hence, there are no three consecutive time steps where $c$ is not occupied by $P_i$. It follows that $P_i$ temporally intersects $P_{i,j'}$.
\end{proof}

\cref{lemma1,lemma2} essentially show that the vertex selection in our reduction works as intended. We obtain the following corollary.

\begin{corollary}\label{corollary1}
    Assume $(\TGcompact, S)$ is a yes-instance of \TempDisjointWalks\ and let $\mathcal{P}$ be a solution. Let $P_i,\tP_i\in\mathcal{P}$ be temporal walks connecting $(s_i,z_i)$ and $(\ts_i,\ttt_i)$, respectively. There exists an $\ell \in [n]$ such that $P_i$ visits $y^i_\ell$ and $\tP_i$ visits $\ty^i_\ell$. Furthermore,  $P_i$ visits no $y^i_{\ell'}$ with $\ell'\neq\ell$ and $\tP_i$ visits no $\ty^i_{\ell'}$ with $\ell'\neq\ell$.
\end{corollary}

Next, we show a property of the temporal walks $P_{i,j}$ connecting $(s_{i,j},z_{i,j})$ for $i<j$ in yes-instances.
\begin{lemma}\label{lemma3}
    Assume $(\TGcompact, S)$ is a yes-instance of \TempDisjointWalks\ and let $\mathcal{P}$ be a solution. Let $P_{i,j}\in\mathcal{P}$ be the temporal walk connecting $(s_{i,j},z_{i,j})$ for $i<j$. There exists an $e=\{u,v\}\in E$ with $u\in V_i$ and $v\in V_j$ such that $P_{i,j}$ has the transitions
    \begin{multicols}{4}
    \begin{itemize}
        \item $(s_{i,j},c,t_1)$,
        \item $(c,\alpha_{u}^j,t_2)$,
        \item $(\alpha_{u}^j,c,t_3)$,
        \item $(c,\beta_e,t_4)$,
        \item $(\beta_e,c,t_5)$,
        \item $(c,\gamma_{v}^i,t_6)$,
        \item $(\gamma_{v}^i, c, t_7)$, 
        \item $(c,z_{i,j},t_8)$,
    \end{itemize}
    \end{multicols}
    for some $t_1<t_2<\ldots<t_8$.
\end{lemma}
\begin{proof}
Temporal walk $P_{i,j}$ (with $i<j$) starts at $s_{i,j}$ and then moves to the center vertex $c$. Hence, the first transition of $P_{i,j}$ is $(s_{i,j},c,t_1)$ for some $t_1$.
In the next time step, $P_{i,j}$ has to move to some vertex $\alpha_u^{j}$ with $u\in V_i$, otherwise $P_{i,j}$ temporally intersects some other temporal walk $P_{i',j'}$ connecting $(s_{i',j'},z_{i',j'})$ or the temporal walk $P_1$ connecting $(s_1,z_1)$. It follows that the second transition of $P_{i,j}$ is $(c,\alpha_{u}^j,t_2)$ for $t_2=t_1+1$ and the third transition of $P_{i,j}$ is $(\alpha_{u}^j,c,t_3)$ for some $t_3>t_2$.

Next, we consider the last three transitions of $P_{i,j}$. The last one is clearly $(c,z_{i,j},t_8)$ for some $t_8$. The previous transition must occur at time $t_7=t_8-1$ since otherwise $P_{i,j}$ temporally intersects some other temporal walk $P_{i',j'}$ connecting $(s_{i',j'},z_{i',j'})$ or the temporal walk $P_k$ connecting $(s_k,z_k)$. It follows that the second-last transition of $P_{i,j}$ is $(\gamma_{v}^i, c, t_7)$ for some $v\in V_j$ and the third-last transition is $(c,\gamma_{v}^i,t_6)$ for some $t_6<t_7$.

It remains to show what happens between the third and the third-last transition of $P_{i,j}$.
Let $P_i,\tP_i\in\mathcal{P}$ be temporal walks connecting $(s_i,z_i)$ and $(\ts_i,\ttt_i)$, respectively. By \cref{corollary1} we have that there exists an $\ell \in [n]$ such that $P_i$ visits $y^i_\ell$ and $\tP_i$ visits $\ty^i_\ell$. Furthermore, we have that $P_i$ visits no $y^i_{\ell'}$ with $\ell'\neq\ell$ and $\tP_i$ visits no $\ty^i_{\ell'}$ with $\ell'\neq\ell$. It follows that, by construction, there is one interval of size $4(k-1)$ where $c$ is neither occupied by $P_i$ nor $\tP_i$ (when the walks move to $y^i_\ell$ and $\ty^i_\ell$, respectively), that we associate with color $i$. Analogously, there is one such interval that we associate with color $j$.
We have that with $(\alpha_{u}^j,c,t_3)$ the temporal walk $P_{i,j}$ enters the interval corresponding to color $i$ and with $(c,\gamma_{v}^i,t_6)$ leaves the interval corresponding to color $j$, otherwise $P_{i,j}$ would temporally intersect at least one of the walks $P_i,\tP_i,P_j,\tP_j$, where $P_j,\tP_j\in\mathcal{P}$ are the temporal walks connecting $(s_j,z_j)$ and $(\ts_j,\ttt_j)$, respectively. We can observe that at most $k-1$ walks $P_{i',j'}\in\mathcal{P}$ connecting $(s_{i',j'},z_{i',j'})$ can enter and leave the interval corresponding to a color. There are $k$ such intervals, there are $\binom{k}{2}$ such walks and each such walk needs to enter at least two such intervals.
By the pigeonhole principle we have that if $P_{i,j}$ enters more than two intervals, that is, the one of color $i$, the one of color $j$, and at least one more of some color $i'$, then there is one interval that is visited by $k$ walks, a contradiction to the assumption that all temporal walks in $\mathcal{P}$ are pairwise temporally disjoint.
It follows that after $P_{i,j}$ leaves the interval of color $i$, the next time it visits the center vertex $c$ needs to be in the interval of color $j$. This is only possible if $e=\{u,v\}\in E$ and $P_{i,j}$ can use transitions $(c,\beta_e,t_4)$ and $(\beta_e,c,t_5)$ with $t_3<t_4<t_5<t_6$. 
\end{proof}

We are now ready to prove \cref{thm:k-star}.
\begin{proof}[Proof of \cref{thm:k-star}]
Let $(\TGcompact, S)$ be the \TempDisjointWalks\ instance described by \cref{constr:star} for \textsc{Multicolored Clique} instance $G=(V_1\uplus\ldots\uplus V_k, E)$. 
The instance $(\TGcompact, S)$ can clearly be computed in polynomial time and we can observe that $|S|\in O(k^2)$. 
We show that $(\TGcompact, S)$ is a yes-instance of \TempDisjointWalks\ if and only if $G$ contains a clique of size $k$.

$(\Rightarrow)$: Assume $(\TGcompact, S)$ is a yes-instance of \TempDisjointWalks\ and let $\mathcal{P}$ be a solution. Consider color $i$ and let $P_i,\tP_i\in\mathcal{P}$ be temporal walks connecting $(s_i,z_i)$ and $(\ts_i,\ttt_i)$, respectively. By \cref{corollary1} we have that there exists an $\ell \in [n]$ such that $P_i$ visits $y^i_\ell$ and $\tP_i$ visits $\ty^i_\ell$. Furthermore,  $P_i$ visits no $y^i_{\ell'}$ with $\ell'\neq\ell$ and $\tP_i$ visits no $\ty^i_{\ell'}$ with $\ell'\neq\ell$. Note that this implies that $\ell$ is uniquely defined. We say that $v_\ell^i$ is the \emph{selected} vertex of color $i$. Now let $X$ be the set of the selected vertices of all colors, that is, $X=\{v_\ell^i\in V_i\mid i\in[k] \text{ and }v_\ell^i\text{ is the selected vertex of color }i\}$. We claim that $X$ is a clique in $G$. Note that by definition, $X$ contains exactly one vertex of each color.

Assume for contradiction that $X$ is not a clique in $G$. Then there exist $v_\ell^i,v_{\ell'}^j\in X$ with $i\neq j$ such that $\{v_\ell^i,v_{\ell'}^j\}\notin E$. Consider the temporal walk $P_{i,j}\in\mathcal{P}$ that connects $(s_{i,j},z_{i,j})$. By \cref{lemma3} we have that there is some $e=\{u,v\}\in E$ with $u\in V_i$ and $v\in V_j$ such that $P_{i,j}$ visits vertices $\alpha_u^j$, $\beta_e$, and $\gamma_v^i$. 
Let $P_i,\tP_i\in\mathcal{P}$ be temporal walks connecting $(s_i,z_i)$ and $(\ts_i,\ttt_i)$, respectively, and let $P_j,\tP_j\in\mathcal{P}$ be temporal walks connecting $(s_j,z_j)$ and $(\ts_j,\ttt_j)$, respectively.
By the construction of~$\TG$ (\cref{constr:star}), we have that the temporal walks $P_{i,j},P_i,\tP_i,P_j,\tP_j$ can only all be pairwise temporally disjoint if $u=v_\ell^i$ and $v=v_{\ell'}^j$, and hence $e=\{v_\ell^i,v_{\ell'}^j\}\in E$, a contradiction. 

$(\Leftarrow)$: Assume $G$ is a yes-instance of \textsc{Multicolored Clique} and let $X\subseteq\bigcup_{i=1}^k V_i$ with $|X\cap V_i|=1$ for all $i\in [k]$ be a multicolored clique in $G$. We construct a solution $\mathcal{P}$ for $(\TG,S)$ as follows.
For each $i \in [k]$, let $a_i \in [n]$ be chosen such that $v^i_{a_i} \in X \cap V_i$.
For each $(s,z) \in S$, we must give a temporal $(s,z)$-walk and then show that these walks are pairwise temporally disjoint.

We start with $(s_{i,j},z_{i,j})$ for $i,j \in [k]$ with $i<j$.
Let $v \coloneqq v^i_{a_i}$, $v'\coloneqq v^j_{a_j}$, and $e \coloneqq \{v,v'\}$.
Observe that, because $v,v' \in X$, the edge $e$ is present in $G$.
The $(s_{i,j},z_{i,j})$-walk $P_{i,j}$ starts in $s_{i,j}$, moves to $c$ in the layer $E^{i,j}_1$, and continues to $\alpha_{v}$ in layer $E^{i,j}_2$.
It remains there until the layer $E^i_{4k(a_i-1)+4(j-1)}$, at which point it returns to $c$.
It then moves on to $\beta_e$ in layer $E^i_{4k(a_i-1)+4(j-1)+2}$.
From there, the walk moves to $c$ in layer $E^j_{4k(a_j-1)+4i}$ and on to $\gamma_{v'}$ in $E^j_{4k(a_j-1)+4i+2}$.
Finally, in the layers $E^{i,j}_{f-1}$ and $E^{i,j}_{f}$, it moves through $c$ to $z_{i,j}$.

Next consider the walks $P_i$ and $\tP_i$ for $i \in [k]$ connecting $(s_i,z_i)$ and $(\ts_i,\ttt_i)$, respectively.
In $E^i_{0}$, the walk $P_i$ moves to from $s_i$ to $c$.
For $\ell \in [(a_i-1)k]$, it moves from $c$ to $w^i_\ell$ in layer $E^i_{4\ell -3}$ and back to $c$ in $E^i_{4\ell}$.
In layer $E^i_{4(a_i-1)k+1}$, the walk then proceeds to $y^i_{a_i}$, remaining there until it returns to $c$ in $E^i_{4a_ik+1}$.
For every $\ell \in [(a_i+1)k,nk]$, the walk moves from $c$ to $x^i_{\ell}$ in $E^i_{4\ell-2}$ and back to $c$ in $E^i_{4\ell+1}$.
In $E^i_{4kn+4}$, the walk reaches its goal $z_i$.
The walk for $\tP_i$ is analogous.

We must now argue that none of the constructed walks temporally intersect.
The only vertex used by more than one walk is $c$, so the constructed walks cannot temporally intersect in any other vertex.
They also cannot temporally intersect in $c$ in a layer $E^{i,j}_r$ with $r \in \{1,2,f-1,f\}$ because only $P_{i,j}$ occupies $c$ in those layers.
In the layers $E^i_{r}$, the vertex $c$ is occupied by the walks $P_i$, $\tP_i$, $P_{i,j}$ with $j\in[i+1,k]$ and $P_{j,i}$ with $j \in [i-1]$.
The walk $P_i$ only occupies $c$ in layers $E^i_{r}$ with
\begin{itemize}
	\item $r \bmod 4 \in \{0,1\}$ and $r \leq 4(a_i-1)k$ or
	\item $r \bmod 4 \in \{1,2\}$ and $r \geq 4(a_i + 1)k +1$.
\end{itemize}
The walk $\tP_i$ only occupies $c$ in layers $E^i_{r}$ with
\begin{itemize}
	\item $r \bmod 4 \in \{2,3\}$ and $r \leq 4(a_i-1)k$ or
	\item $r \bmod 4 \in \{0,4\}$ and $r \geq 4(a_i + 1)k +1$.
\end{itemize}
Hence, those two also do not temporally intersect.
Meanwhile, the walks $P_{i,j}$ and $P_{j,i}$ only occupy $c$ in the layers $E^i_{4k(a_i-1)+4+(j-1)},\ldots,E^i_{4k(a_i-1)+4+(j-1)+2}$, so they also do not intersect pairwise or with $P_i$ and $\tP_i$.
\end{proof}

\section{Algorithms for \AllProbs}
In this section, we present two new algorithms, one for \TempDisjointPaths\  and one for \TempDisjointWalks.
For \TempDisjointPaths, we present in \cref{sec:FESalgo} an FPT-algorithm for the combination of the number of source-sink pairs and the feedback edge number of the underlying graph as a parameter. For \TempDisjointWalks, we give  in \cref{sec:walksonlines} an FPT-algorithm for the number of source-sink pairs that requires the underlying graph of the input to be a path.

In both cases we generalize results and resolve open questions by \citet{KlobasMMNZ23}. Furthermore, our computational hardness results in \cref{sec:hardness} imply that our algorithmic result cannot be improved significantly (from a classification standpoint).

\subsection{Algorithm for \tdp}\label{sec:FESalgo}
In this section, we present an FPT-algorithm for \TempDisjointPaths\ parameterized by the combination of the number of source-sink pairs and the feedback edge set number of the underlying graph. This generalizes the FPT-algorithm by \citet{KlobasMMNZ23} for \TempDisjointPaths\ parameterized by the number of source-sink pairs for temporal forests. \cref{thm:w1hardnessVCN} implies that we presumably cannot replace the feedback edge set number of the underlying graph by a smaller parameter such as the feedback vertex number or the treewidth and still obtain fixed-parameter tractability.

\begin{theorem}\label{thm:fes}
\tdp\ is fixed-parameter tractable when parameterized by the combination of the number $|S|$ of source-sink pairs and the feedback edge number of the underlying graph.    
\end{theorem}

The high-level idea of the algorithm is as follows.
We can bound the number of paths in the underlying graph between any source-sink pair in a function of its feedback edge number. We can do the same for the number of how often two such paths intersect. Hence, for a given set of paths, the total number of such intersections is bounded by a function of the feedback edge set number and the number of source-sink pairs. For each such intersection, we can consider all possibilities in which order it is traversed by the temporal paths. This gives us enough information to verify in polynomial time whether the possibility of how and in which order the source-sink pairs should be connected is realizable.

\begin{proof}[Proof of \cref{thm:fes}]
Let $(\TGcompact, S)$ be an instance of \TempDisjointPaths\ such that the underlying graph of $\TG$ is a path. Let $G$ denote the underlying graph of $\TG$. Recall that $\hat{S}$ denotes the set of all vertices in $V$ that appear as sources or sinks in $S$.

Since we only consider temporal paths in the solution, we can perform the following preprocessing step. We exhaustively remove all vertices $v\in V$ that have degree at most one in the underlying graph $G$ and that are not sources or sinks, that is, $v\notin \hat{S}$. A temporal $(s,z)$-path for $(s,z)\in S$ cannot visit such vertices. Let $G'$ denote the underlying graph after the preprocessing step and let $F$ be a minimum feedback edge set of $G'$. We call a vertex $v$ of $G'$ \emph{interesting} if one of the following properties hold.
\begin{itemize}
    \item Vertex $v$ is a source or a sink, that is, $v\in\hat{S}$.
    \item Vertex $v$ is incident with a feedback edge.
    \item Vertex $v$ has degree at least three in $G'$.
\end{itemize}
Let $D$ denote the set of interesting vertices. We have that there are $O(|F|+|S|)$ interesting vertices in $G'$, that is, $|D|\in O(|F|+|S|)$. Furthermore, there are $O(|F|+|S|)$ \emph{interesting path segments} in $G'$,  that is, paths in $G'$ start and end at interesting vertices and have no internal vertices that are interesting~\cite[Lemma~2]{BentertDKNN20}. Let $\mathcal{P}$ denote the set of interesting path segments. Clearly, every temporal $(s,z)$-path for $(s,z)\in S$ follows a path in $G'$ that is composed of interesting path segments. For each $(s,z)\in S$ there are $2^{O(|F|+|S|)}$ different set of interesting path segments that are visited by a temporal $(s,z)$-path. For a temporal $(s,z)$-path $P$, we call the set of interesting path segments that are visited by $P$ the \emph{configuration} of $P$. Note that there is one unique order and direction for $P$ to traverse the interesting path segments in its configuration. 
Considering all source-sink pairs, we have $2^{O(|F|\cdot |S|+|S|^2)}$ possible sets of configurations for the temporal paths in a solution.

For each such set of configurations we consider in which order the temporal $(s,z)$-paths with $(s,z)\in S$ occupy the interesting vertices $D$. 
For a fixed set of configurations there are $|S|^{O(|S|)}$ possible orderings for one fixed interesting vertex. Hence, since $|D|\in O(|F|+|S|)$, there are
$|S|^{O(|F|\cdot |S|+|S|^2)}$ possible sets of orderings for all interesting vertices. However, some of these orderings may not be compatible with each other, that is, they are not realizable without making the temporal $(s,z)$-paths temporally intersect. 
Consider a fixed set of configurations. We add the following constraint to the set of orderings.
\begin{itemize}
    \item 
    If the two temporal paths $P,P'$ connecting $(s,z)\in S$ and $(s',z')\in S$, respectively, have the same interesting path segment in their configuration, then they must have the same relative order in both endpoints of the path segment.
\end{itemize}
If both temporal paths traverse the interesting path segment in the same direction, they cannot ``overtake'' each other without temporally intersecting, that is, if the order in which they occupy the first vertex of the segment is the reversed order in which they occupy the last vertex of the segment, then there must be a vertex in the path segment where the two paths temporally intersect. Similarly, if they traverse the path segment in opposite directions the earlier temporal path must finish traversing the segment before the later temporal path can start, otherwise they would temporally intersect. This is illustrated in \cref{fig:condition}.

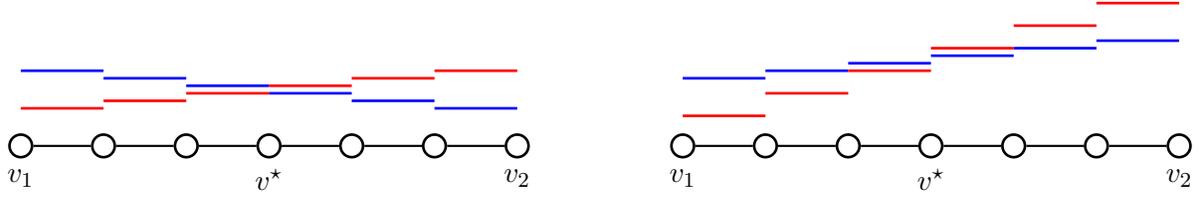
\begin{figure}[t]
    \centering
    \begin{tikzpicture}[line width=1pt, xscale=.55]
		\path
			(0,2) node[vert,label=below:$v_1$](b0) {}
			(2,2) node[vert](b1) {}
			(4,2) node[vert](b2) {}
			(6,2) node[vert,label=below:$v^\star$](b3) {}
			(8,2) node[vert](b4) {}
			(10,2) node[vert](b5) {}
			(12,2) node[vert,label=below:$v_2$](b6) {};
		\draw[edge] (b0) edge node [midway, below] {} (b1);
		\draw[edge] (b1) edge node [midway, below] {} (b2);
		\draw[edge] (b2) edge node [midway, below] {} (b3);
		\draw[edge] (b3) edge node [midway, below] {} (b4);
		\draw[edge] (b4) edge node [midway, below] {} (b5);
		\draw[edge] (b5) edge node [midway, below] {} (b6);
		\draw[blue] (10,2.5) -- (12,2.5);
		\draw[blue] (8,2.6) -- (10,2.6);
		\draw[blue] (6,2.7) -- (8,2.7);
		\draw[blue] (4,2.8) -- (6,2.8);
		\draw[blue] (2,2.9) -- (4,2.9);
		\draw[blue] (0,3) -- (2,3);
    
        \draw[red] (0,2.5) -- (2,2.5);
		\draw[red] (2,2.6) -- (4,2.6);
		\draw[red] (4,2.7) -- (6,2.7);
		\draw[red] (6,2.8) -- (8,2.8);
		\draw[red] (8,2.9) -- (10,2.9);
		\draw[red] (10,3) -- (12,3);

  		\path
			(16,2) node[vert,label=below:$v_1$](c0) {}
			(18,2) node[vert](c1) {}
			(20,2) node[vert](c2) {}
			(22,2) node[vert,label=below:$v^\star$](c3) {}
			(24,2) node[vert](c4) {}
			(26,2) node[vert](c5) {}
			(28,2) node[vert,label=below:$v_2$](c6) {};
		\draw[edge] (c0) edge node [midway, below] {} (c1);
		\draw[edge] (c1) edge node [midway, below] {} (c2);
		\draw[edge] (c2) edge node [midway, below] {} (c3);
		\draw[edge] (c3) edge node [midway, below] {} (c4);
		\draw[edge] (c4) edge node [midway, below] {} (c5);
		\draw[edge] (c5) edge node [midway, below] {} (c6);
		\draw[blue] (16,2.9) -- (18,2.9);
		\draw[blue] (18,3) -- (20,3);
		\draw[blue] (20,3.1) -- (22,3.1);
		\draw[blue] (22,3.2) -- (24,3.2);
		\draw[blue] (24,3.3) -- (26,3.3);
		\draw[blue] (26,3.4) -- (28,3.4);
    
        \draw[red] (16,2.4) -- (18,2.4);
		\draw[red] (18,2.7) -- (20,2.7);
		\draw[red] (20,3) -- (22,3);
		\draw[red] (22,3.3) -- (24,3.3);
		\draw[red] (24,3.6) -- (26,3.6);
		\draw[red] (26,3.9) -- (28,3.9);
\end{tikzpicture}
    \caption{Illustration of an interesting path segments from interesting vertex $v_1$ to interesting vertex~$v_2$ and two temporal paths that follow the segment (red and blue). The horizontal position of the colored edges represents the time labels. Vertex $v_1$ is occupied first be the red path and then by the blue path. Vertex $v_2$ is occupied first by the blue path and then by the red path. On both sides, the red temporal path follows the interesting path segment from $v_1$ to $v_2$. On the right side, the blue temporal path also follows the segment from $v_1$ to $v_2$, and on the left side, the blue temporal path follows the segment from $v_2$ to $v_1$. We can see that in both cases, the paths temporally intersect in $v^\star$.}
    \label{fig:condition}
\end{figure}

Furthermore, we can consider each set of orderings as a partial order $<_1$ over $D\times S$, that is, all combinations of an interesting vertex and a source-sink pair.
For each source-sink pair $(s,z)\in S$ the configuration of temporal $(s,z)$-path implicitly defines an ordering in which the interesting vertices are visited by the path. Considering those orders for all source-sink pairs, we get another partial order $<_2$ over $D\times S$.
Let $<^\star$ denote the transitive closure of the relation $<_{12}$ over $D\times S$ with $x<_{12}y$ if $x<_{1}y$ or $x<_{2}y$. If $<^\star$ is a partial order, that is, there are are no $x,y\in D\times S$ such that $x <^\star y$ and $y <^\star x$, then $<_1$ and $<_2$ have a common linearization which we call the \emph{iteration order}.
If this is the case and the set of orderings also meets the above constraint, then we call the set of orderings \emph{valid}. Note that we can check in polynomial time whether a given set of orderings is valid and if so compute an iteration order in polynomial time.

Having a configuration together with a valid set of orderings, we show that we can check in polynomial time whether it is realizable by pairwise temporally disjoint $(s,z)$-paths for all $(s,z)\in S$.
Formally, the algorithm executes the following steps.
\begin{enumerate}
	\item Exhaustively remove vertices with degree $\le 1$ from the underlying graph $G$, except sources or sinks, that is, vertices in $\hat{S}$. Let~$G'$ be the resulting (static) graph.
	\item Compute a minimum feedback edge set $F$ of $G'$. 
	\item Let $V^{\ge 3}$ denote all vertices of $G'$ with degree at least three. 
		Partition the forest $G' - F$ into a set of maximal paths $\mathcal{P}$ with endpoints in $D=\bigcup_{e\in F} e \cup V^{\ge 3} \cup \hat{S}$, and intermediate vertices all of degree 2.
		It holds that $|\mathcal{P}| \in \mathcal{O}(|F|+|S|)$~\cite[Lemma~2]{BentertDKNN20}.

  Here, $D$ is the set of interesting vertices and $\mathcal{P}$ is the set of interesting path segments.
	\item 
		Iterate over all configurations for all temporal $(s,z)$-paths with $(s,z)\in S$.
  \item Consider a fixed configuration for all temporal $(s,z)$-paths with $(s,z)\in S$.
  
  Iterate over all valid sets of orderings in which the temporal $(s,z)$-paths in the solution occupy the vertices in $D$.

  \item Consider a fixed valid set of orderings. Compute an iteration order for $D\times S$ and iterate over the combinations of interesting vertices and source-sink pairs in that order. Let $(v, (s,z))$ with $v\in D$ and $(s,z)\in S$ be the current combination. 
  \begin{itemize}
      \item If $v=s$, then let $v^+\in D$ denote the next interesting vertex that is visited by the temporal path connecting $s$ and $z$. Let $v'$ be the vertex in the path segment from $v$ to $v^+$ that is right before $v^+$. 
      Compute a prefix-foremost temporal path $P$ (earliest possible arrival time at every vertex) from $v$ to $v'$ that follows the path segment from $v$ to $v^+$. 

    If no such temporal path exists, discard the current combination of configuration with valid set of orderings.
    Otherwise, for each transition $(v,w,t)$ in $P$, remove all time edges incident with $v$ or $w$ that have a time label $t'\le t$.
      %
    \item If $v=z$, then let $v^-\in D$ denote the previous interesting vertex that is visited by the temporal path connecting $s$ and $z$. Let $v'$ be the vertex in the path segment from $v^-$ to $v$ that is right before $v$. Let $(\{v',v\},t)$ be the time edge between $v'$ and $v$ with the smallest time label.

    If no such time edge exists, discard the current combination of configuration with valid set of orderings.
    Otherwise, remove all time edges incident with $v'$ or $v$ that have a time label $t'\le t$.
    \item If $s\neq v\neq z$, then let $v^+\in D$ denote the next interesting vertex that is visited by the temporal path connecting $s$ and $z$ and let $v^-\in D$ denote the previous interesting vertex that is visited by the temporal path connecting $s$ and $z$. Let $v_1$ be the vertex in the path segment from $v^-$ to $v$ that is right before $v$ and let $v_2$ be the vertex in the path segment from $v$ to $v^+$ that is right before $v^+$. 
      Compute a prefix-foremost temporal path $P$ (earliest possible arrival time at every vertex) from $v_1$ to $v_2$ that visits $v$ right after $v_1$ and then follows the path segment from $v$ to $v^+$. 

    If no such temporal path exists, discard the current combination of configuration with valid set of orderings.
    Otherwise, for each transition $(v,w,t)$ in $P$, remove all time edges incident with $v$ or $w$ that have a time label $t'\le t$.
  \end{itemize}
  Proceed with the next combination of interesting vertex and source-sink pair. If there is no further combination, output YES.
      \item If all combinations of configuration with valid set of orderings were discarded, output NO.
\end{enumerate}
 Since we have $2^{O(|F|\cdot |S|+|S|^2)}$ possible configurations and $|S|^{O(|F|\cdot |S|+|S|^2)}$ possible valid sets of orderings, the running time of the algorithm is in $2^{O((|F|\cdot |S|+|S|^2)^2\cdot\log |S|)}\cdot |\TG|^{O(1)}$. 
Note that with polynomial overhead, the algorithm can also output the solution.

By the arguments made before, it is easy to check that if the algorithm outputs YES, then we face a yes-instance. 

For the other direction, assume that we face a yes-instance and thus there is a solution $\mathcal{S}$. For each $(s,z)\in S$, the corresponding temporal $(s,z)$-path in $\mathcal{S}$ follows a path in the underlying graph $G$. Since the temporal paths do not revisit vertices, all degree one vertices (that are not sources or sinks) can be exhaustively removed. It follows that each temporal $(s,z)$-path in $\mathcal{S}$ follows a path in $G'$. Consider the corresponding paths in $G'$. Each of them can be segmented into interesting path segments, where the set of segments implicitly defines the order and direction in which they are traversed. Hence, we can assume that in some iteration of the algorithm, we are considering for each $(s,z)\in S$ the set of interesting path segments that is traversed by the temporal $(s,z)$-path in $\mathcal{S}$, that is, the path's configuration. Naturally, we also have that every interesting vertex is occupied by (a subset of) the temporal $(s,z)$-paths in $\mathcal{S}$ in a certain order. 
By the arguments in the description of the algorithm, we have that if the two temporal paths in $\mathcal{S}$ connecting $(s,z) \in S$ and $(s',z') \in S$ traverse the same interesting path segment, then they must have the same relative order in both endpoints of the path segment. Hence, we have that the set of orderings is valid and will be considered by the algorithm. 

Lastly, note that we can assume w.l.o.g.\ that the temporal $(s,z)$-paths in $\mathcal{S}$ traverse the interesting path segments in a prefix-foremost (among the ones that do not temporally intersect) way up until the second-last vertex in the segment, and the last vertex is visited at the earliest possible time.
If they are not, we can simply replace the temporal path segments (up to the second last vertex) with prefix-foremost ones and replace the last transition with the earliest possible one. We can conclude that the algorithm outputs YES. 
\end{proof}





\subsection{Algorithm for \TempDisjointWalks}\label{sec:walksonlines}
In this section, we present an FPT-algorithm for \TempDisjointWalks\ parameterized by the number of source-sink pairs for the case where the underlying graph is a path.
Recall that a temporal graph that has a path as underlying graph is called a \emph{temporal line}.
\citet{KlobasMMNZ23} showed that \tdw{} is \NP-hard on temporal lines and they gave an FPT-algorithm for \tdp\ parameterized by the number of source-sink pairs for temporal forests. \cref{thm:k-star} implies that we presumably cannot adapt this FPT-algorithm for \tdw. However, we can obtain tractability for the case of temporal lines. This answers an open question by \citet{KlobasMMNZ23}.

\begin{theorem}
	\label{thm:k-line}
	\tdw{} on temporal lines is \fpt{} with respect to the number $|S|$ of source-sink pairs.
\end{theorem}

Before we prove \cref{thm:k-line}, we first investigate properties of solutions $\mathcal{S}$ to an instance of \tdw{} that minimize the sum of the lengths of its walks (our algorithm will produce such a solution).
We show that we can upper-bound the number of times a temporal walk in such a solution $\mathcal{S}$ changes its direction by a function of $|S|$. Furthermore, we show that the direction changes always occur in ``regions'' (whose size is upper-bounded by a function of $|S|$) ``around'' the sources and sinks in $S$. Intuitively speaking, this allows us to iterate over all possibilities in which direction, how often, and in which order the temporal walks connecting the source-sink pairs move through the regions around the source and sink vertices in $S$. Given such a possibility, we have enough information to check whether there exist temporally disjoint walks that realize this behavior.

For the remainder of this section, let $\mathcal{S}$ be a solution to an instance of \tdw{} that minimizes the sum of the lengths of its temporal walks. We first show that if a temporal walk in~$\mathcal{S}$ changes its direction, there has to be another temporal walk in~$\mathcal{S}$ that enforces this behavior as follows.

\begin{lemma}\label{claim:fptwalks:1}
Let $W$ be a temporal $(s,z)$-walk in $\mathcal{S}$ such that $(a,b,t),(b,a,t')$ with $t<t'$ are consecutive in~$W$. 
Then, there exists a temporal $(s',z')$-walk $W'$ in $\mathcal{S}$ with $(c,a,t'')$ or $(a,c,t'')$ 
in $W'$ where $t<t''<t'$.
\end{lemma}
\begin{proof}
Let $(a',a,t_{-1})$ be the transition in $W$ before $(a,b,t)$ and let $(a,a'',t'_{+1})$ be the transition in $W$ after $(b,a,t')$. Now consider $\hat{W}$ that is obtained from $W$ by removing the transitions $(a,b,t),(b,a,t')$. Clearly, $\hat{W}$ is a temporal $(s,z)$-walk and its length is smaller than the length of $W$.  Since $\hat{W}\notin \mathcal{S}$, there must be a $W'\in \mathcal{S}$ that temporally intersects with $\hat{W}$ but not with $W$. Comparing $W$ with $\hat{W}$, the only vertex that is occupied by $\hat{W}$ for a longer time than by $W$ is vertex $a$. The temporal walk $W$ occupies $a$ from time $t_{-1}$ to time $t$ and then again from time $t'$ to time $t'_{+1}$ (with $t_{-1}<t<t'<t'_{+1}$), whereas the temporal walk $\hat{W}$ occupies vertex $a$ from time $t_{-1}$ to time $t'_{+1}$. Hence, the time period where $a$ is occupied by $\hat{W}$ but not by $W$ is $(t,t')$. It follows that $W'$ has to occupy $a$ at least once at some time $t''$ with $t<t''<t'$. Assume that $\hat{W}$ arrives at $a$ at time $t''$, then $\hat{W}$ has to contain a transition $(c,a,t'')$ for some vertex $c$. If $\hat{W}$ never arrives at $a$, then it must start at $a$ and contain a transition $(a,c,t'')$ for some $c$.
\end{proof}

Having \cref{claim:fptwalks:1}, we can inductively show if a temporal walk $W_0$ in $\mathcal{S}$ changes direction, then there is a sequence of temporal walks in $\mathcal{S}$ that change direction right before $W_0$ followed by a temporal walk in $\mathcal{S}$ that is either starting at its source or arriving at its sink.

\begin{lemma}\label{claim:fptwalks:2}
Let $W_0$ be a temporal $(s_0,z_0)$-walk in $\mathcal{S}$ such that $(a_0,b_0,t_0),(b_0,a_0,t_0')$ with $t_0<t'_0$ are consecutive in $W_0$. 
Then, there exist temporal $(s_1,z_1)$-walk $W_1$, \ldots, temporal $(s_r,z_r)$-walk $W_r$ in $\mathcal{S}$ and $a_1,b_1,t_1,t_1',\ldots,a_r,b_r,t_r,t_r'$ so that:
\begin{itemize}
    \item For every $1\le i<r$, $(a_i,b_i,t_i),(b_i,a_i,t_i')$ are consecutive in $W_i$, 
    $t_{i-1}<t_i\le t_i'<t_{i-1}'$ and $a_{i-1}=b_i$.
    \item either $(a_r,b_r,t_r)$ or $(b_r,a_r,t_r)$ in $W_r$, $t_{r-1}<t_r<t_{r-1}'$, $a_{r-1}=b_r$ and either $b_r=z_r$ or $b_r=s_r$.
\end{itemize}
\end{lemma}
\begin{proof}
Consider $W_0$. We prove the statement by induction on $t_0'-t_0$, that is, the time difference between when $W_0$ leaves vertex $a_0$ and when it comes back. 
 By \cref{claim:fptwalks:1} there exists an $(s',z')$-walk $W'$ in $\mathcal{S}$ with $(c,a_0,t'')$ or $(a_0,c,t'')$ 
in $W'$ where $t_0<t''<t_0'$. If $a_0=s'$ or $a_0=z'$, then we have $r=1$ and $W'=W_r$ and the statement holds. 

Otherwise, we have the following. Assume that $s'\neq a_0\neq z'$. Rename $c=a_1$, $a_0=b_1$, $s'=s_1$, $z'=z_1$, and $W'=W_1$. By \cref{claim:fptwalks:1} we have that $(a_1,b_1,t_1)$ or $(b_1,a_1,t_1)$ with $t_0<t_1<t_0'$ are transitions of $W_1$. However, since $s_1\neq b_1=a_0\neq z_1$, we have that both transitions must be consecutive in $W_1$, where $(a_1,b_1,t_1)$ is the first and $(b_1,a_1,t_1')$ is the second for some $t_0'>t_1'>t_1>t_0$, otherwise $W_1$ would temporally intersect with $W_0$.

Now $W_1$ is an $(s_1,z_1)$-walk in $\mathcal{S}$ such that $(a_1,b_1,t_1),(b_1,a_1,t_1')$ are consecutive in $W_1$, with $t_1'-t_1<t_0'-t_0$. Hence, by induction we know that there exist $(s_2,z_2)$-walk $W_2$, \ldots, $(s_r,z_r)$-walk $W_r$ in $\mathcal{S}$ and $a_2,b_2,t_2,t_2',\ldots,a_r,b_r,t_r,t_r'$ so that:
\begin{itemize}
    \item For every $2\le i<r$, $(a_i,b_i,t_i),(b_i,a_i,t_i')$ are consecutive in $W_i$ and $b_i$ is closer to $u_i$ than $a_i$, $t_{i-1}<t_i\le t_i'<t_{i-1}'$ and $a_{i-1}=b_i$.
    \item either $(a_r,b_r,t_r)$ or $(b_r,a_r,t_r)$ in $W_r$, $t_{r-1}<t_r<t_{r-1}'$, $a_{r-1}=b_r$ and either $b_r=z_r$ or $b_r=s_r$.
\end{itemize}
Hence, the lemma follows.
\end{proof}

From \cref{claim:fptwalks:2} we can draw two important corollaries that will help us to design the algorithm and prove its correctness. We give an illustration in \cref{fig:lemma}. The first corollary is that direction changes of temporal walks in $\mathcal{S}$ occur not too far away from source or sink vertices.

\begin{figure}[t]
    \centering
    \begin{tikzpicture}[line width=1pt, scale=1.2]
		\path
			(0,2) node[vert,label=below:$a_r$](b0) {}
			(2,2) node[vert, diamond, fill=black,label=below:$b_r\eq a_{r-1}$](b1) {}
			(4,2) node[vert,label=below:$b_{r-1}\eq a_{r-2}$](b2) {}
			(6,2) node[vert,label=below:$b_{r-2}\eq a_{r-3}$](b3) {};
   \path
			(10,2) node[vert,label=below:$b_1\eq a_0$](b5) {}
			(12,2) node[vert,label=below:$b_0$](b6) {};
		\draw[edge] (b0) edge node [midway, below] {} (b1);
		\draw[edge] (b1) edge node [midway, below] {} (b2);
		\draw[edge] (b2) edge node [midway, below] {} (b3);
		\draw[edge, dashed] (b3) edge node [midway, below] {} (b5);
		\draw[edge] (b5) edge node [midway, below] {} (b6);

		\draw[blue] (0,3) -- (2,3);
    
        \draw[red,dashed] (0,2.8) -- (2,2.8);
		\draw[red] (2,2.9) -- (2,2.8);
		\draw[red] (2,2.9) -- (4,2.9);
		\draw[red] (4,3.1) -- (4,2.9);
		\draw[red] (4,3.1) -- (2,3.1);
		\draw[red] (2,3.2) -- (2,3.1);
		\draw[red,dashed] (2,3.2) -- (0,3.2);

          \draw[olive,dashed] (2,2.6) -- (4,2.6);
		\draw[olive] (4,2.8) -- (4,2.6);
		\draw[olive] (4,2.8) -- (6,2.8);
		\draw[olive] (6,3.2) -- (6,2.8);
		\draw[olive] (6,3.2) -- (4,3.2);
		\draw[olive] (4,3.4) -- (4,3.2);
		\draw[olive,dashed] (4,3.4) -- (2,3.4);

            \draw[violet,dashed] (4,2.4) -- (6,2.4);
            \draw[violet,dashed] (6,2.6) -- (6,2.4);
            \draw[violet,dashed] (6.3,2.6) -- (6,2.6);
		\draw[violet,dashed] (6.3,3.4) -- (6,3.4);
		\draw[violet,dashed] (6,3.6) -- (6,3.4);
		\draw[violet,dashed] (6,3.6) -- (4,3.6);

		\draw[brown,dashed] (7.7,2.4) -- (8,2.4);
		\draw[brown,dashed] (8,2.7) -- (8,2.4);
		\draw[brown] (8,2.7) -- (10,2.7);
		\draw[brown] (10,3.3) -- (10,2.7);
		\draw[brown] (10,3.3) -- (8,3.3);
		\draw[brown,dashed] (8,3.6) -- (8,3.3);
		\draw[brown,dashed] (8,3.6) -- (7.7,3.6);

            \draw[orange,dashed] (8,2.2) -- (10,2.2);
		\draw[orange] (10,2.5) -- (10,2.2);
		\draw[orange] (10,2.5) -- (12,2.5);
		\draw[orange] (12,3.5) -- (12,2.5);
		\draw[orange] (12,3.5) -- (10,3.5);
		\draw[orange] (10,3.8) -- (10,3.5);
		\draw[orange,dashed] (10,3.8) -- (8,3.8);
\end{tikzpicture}
    \caption{Illustration for \cref{claim:fptwalks:2} and \cref{claim:fptwalks:3}. The colored paths represent temporal $(s_r,z_r)$- to $(s_0,z_0)$-walks. The horizontal position of the colored edges represents the time labels. The diamond shaped vertex $b_r=a_{r-1}$ equals $s_r$ or $z_r$ and hence is a vertex of interest. Since $r\le |S|$, the distance between $b_r$ and $b_0$ is also at most $|S|$.}
    \label{fig:lemma}
\end{figure}
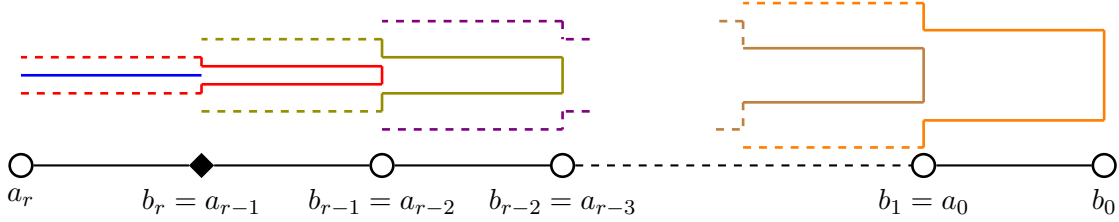

\begin{corollary}\label{claim:fptwalks:3}
Let $W$ be an $(s,z)$-walk in $\mathcal{S}$ such that $(a,b,t),(b,a,t')$ are consecutive in $W$ with $t<t'$. 
Then, there exists an $(s',z')$-walk $W'\neq W$ in $\mathcal{S}$ with at least one among $(s',a',t'')$ or $(a',z',t'')$ in $W'$ for some $a'$, where $t<t''<t'$ and the distance between $a$ and $a'$ in the underlying path of the temporal graph is at most $|S|$.
\end{corollary}

The second corollary is that the temporal walks in $\mathcal{S}$ do not change their direction too often.

\begin{corollary}\label{claim:fptwalks:4}
Let $W$ be an $(s,z)$-walk in $\mathcal{S}$. Then at most $2|S|$ pairs of triples of the form $(a,b,t),(b,a,t')$ with $t<t'$ are consecutive in $W$. 
\end{corollary}


We now have all the pieces we need to prove \cref{thm:k-line}.

\begin{proof}[Proof of \cref{thm:k-line}]
Let $(\TGcompact, S)$ be an instance of \TempDisjointWalks\ such that the underlying graph of $\TG$ is a path. Let $G$ denote the underlying graph of $\TG$. Recall that $\hat{S}$ denotes the set of all vertices in $V$ that appear as sources or sinks in $S$. We know by \cref{claim:fptwalks:3} that we may assume w.l.o.g.\ that all temporal walks in a solution to $(\TG, S)$ change direction only at vertices that are of distance (in $G$) at most $|S|$ from some vertex $v\in \hat{S}$. Let $D=\{v\in V\mid \exists v'\in \hat{S} \text{ such that } \text{dist}_G(v,v')\le |S|\}$ denote the set of all vertices in $V$ where some temporal walk in the solution potentially changes direction. Observe that $|D|\le 4|S|^2$. By \cref{claim:fptwalks:4} we know that w.l.o.g.\ all temporal walks in a solution to $(\TG,S)$ change direction at most $2|S|$ times. It follows that for each source-sink pair in $(s,z)\in S$, there are $|S|^{O(|S|)}$ possibilities we need to consider for where the temporal $(s,z)$-walk in the solution changes directions. Considering all source-sink pairs, we have $|S|^{O(|S|^2)}$ possible configurations the we need to consider for where the temporal walks in the solution change directions.

Consider one specific configuration. We now analyse how many different relative orderings of the temporal walks we need to consider. To do this, we treat every temporal walk as at most $2|S|$ temporal path segments that form the walk, that is, at the endpoints of each path segment, the walk changes direction (or starts/ends). In total, this gives us $2|S|^2$ temporal path segments. Note that any two of these path segments $P,P'$ have to property that they either do not visit common vertices, or if they do, then for all common vertices we have that either $P$ occupies each of them before $P'$ or vice versa. Otherwise, $P$ and $P'$ would be temporally intersecting. It follows that there exist a total ordering of all path segments such that whenever two path segments visit common vertices, the ordering defines which of the two path segments occupies each of the common vertices first.
Overall, we have $|S|^{O(|S|^2)}$ possible orderings for the path segments.

However, some of these orderings might not yield pairwise temporally disjoint walks when we reconnect all path segments to form the respective temporal walks. Let $P_1$ and $P_2$ be two consecutive path segments of some temporal walk $W$ such that $P_1$ is the path segment right before $P_2$, then the ordering must obey two requirements.
\begin{enumerate}
    \item Path segment $P_1$ must occur before path segment $P_2$ in the ordering.
    \item Let vertex $v$ be the endpoint of $P_1$ and the starting point of $P_2$, implying that $W$ changes direction at vertex $v$. Then for each path segment $P'$ that contains vertex $v$ we have that $P'$ either must be before $P_1$ and $P_2$ in the ordering or $P'$ must be after $P_1$ and $P_2$ in the ordering.
\end{enumerate}
The first requirement must be met, since otherwise connecting $P_1$ and $P_2$ does not yield a temporal walk. To see why the second requirement must be met, let $W'$ be the temporal walk of which $P'$ is a path segment. If $W=W'$, then the first requirement is not met. If $W\neq W'$, then the two temporal walks would temporally intersect in vertex $v$.

We call an ordering of the path segments \emph{valid} if both the above requirements are met. Given an ordering, we can clearly check in polynomial time whether it is valid or not.

The algorithm now proceeds as follows:
\begin{enumerate}
    \item Iterate over all possible configurations for where the temporal walks in the solution change directions.
    \item For each configuration, iterate over all valid orderings of the path segments implicitly given by the configuration.
    \item For each configuration with a valid ordering, iterate over the path segments according to the ordering.

    For each path segment, compute a prefix-foremost temporal path $P$ (earliest possible arrival time at every vertex) from the starting point to the endpoint of the path segment. 
    
    If no such temporal path exists, discard the current combination of configuration with valid ordering.
    Otherwise, for each transition $(v,w,t)$ in $P$, remove all time edges incident with $v$ or $w$ that have a time label $t'\le t$. Continue with the next path segment. If there is no further path segment, output YES.
    \item If all combinations of configuration with valid ordering were discarded, output NO.
\end{enumerate}
Since we have $|S|^{O(|S|^2)}$ possible configurations and $|S|^{O(|S|^2)}$ possible valid orderings, the running time of the algorithm is in $|S|^{O(|S|^4)}\cdot |\TG|^{O(1)}$. 
Note that with polynomial overhead, the algorithm can also output the solution.

By the arguments made before, it is easy to check that if the algorithm outputs YES, then we face a yes-instance. 

For the other direction, assume that we face a yes-instance. Then there is a solution $\mathcal{S}$ that minimizes the sum of the lengths of its temporal walks. \cref{claim:fptwalks:3,claim:fptwalks:4} imply that each temporal walk in the solution changes direction at most $2|S|$ times at vertices that are of distance at most $|S|$ to a vertex that appears as a source or sink in $S$. Hence, we can  segment every temporal walk in the solution into at most $2|S|$ temporal paths that have endpoints in the set $D$ (defined at the beginning of the proof). The path segments form a partially ordered set, if we define a path segment $P$ to be smaller than $P'$ if the two path segments have common vertices and each of the common vertices is occupied by $P$ earlier than by $P'$. Note that for all pairs of path segment $P,P'$ that have common vertices, we have that $P$ is either smaller than $P'$ or vice versa, otherwise $P$ and $P'$ would be temporally intersecting. Hence, we have that any linearization of the partial ordering is a valid ordering of the path segments.

If follows that there exists a combination of configuration with valid ordering that agrees with the solution. Lastly, note that we can assume w.l.o.g.\ that the temporal path segments in the solution are prefix-foremost (among the ones that do not temporally intersect), since if they are not, we can simply replace a temporal path segment with a prefix-foremost one. We can conclude that the algorithm outputs YES. 
\end{proof}

\section{Conclusion}\label{sec:conclusion}
Building upon the work of \citet{KlobasMMNZ23}, we presented an almost complete picture of the parameterized computational complexity of \TempDisjointPaths\ and \TempDisjointWalks\ when structural graph parameters of the underlying graph combined with the number of source-sink pairs are considered. For both problem variants, we showed W[1]-hardness for the number of vertices as the parameter even for instances where the underlying graph is a star, indicating that solely restricting the structure of the underlying graph is insufficient for obtaining tractability. Consequently, we considered the number of source-sink pairs as an additional parameter. For \TempDisjointPaths\ we showed that even combining the number of source-sink pairs with the vertex cover number of the underlying graph is presumably insufficient to obtain fixed-parameter tractability. However, we showed that the problem is in FPT when parameterized by the feedback edge number combined with the number of source-sink pairs. For \TempDisjointPaths\ we showed that the problem remains W[1]-hard when parameterized by the number of source-sink pairs even if the underlying graph is a star. However, the problem is in FPT for the same parameter if the underlying graph is a path. Our results revealed surprising differences between the path- and the walk-variant of our problem and resolved open questions by \citet{KlobasMMNZ23}.

We leave several directions for future research. Our hardness results rule out many structural graph parameters as further options for obtaining tractability. However, there are some candidates that are unrelated to the vertex cover number and the feedback edge set, and are also large on star graphs, leading to the following question.
\begin{itemize}
    \item Is \AllProbs\ in FPT or W[1]-hard when parameterized by the combination of the number of source-sink pairs and the \emph{cutwidth} or \emph{bandwidth} of the underlying graph?
\end{itemize}

In MAPF, one is often interested in finding solutions that minimize the sum or maximum of steps or actions that each agent needs to take to arrive at their destination~\cite{Stern19}. In our setting, the number of transitions of a temporal path or walk corresponds to the number of steps and the difference between the time label of the last and first transition (also called \emph{duration}) corresponds to the number of actions (where waiting for one time step is considered an action). We can observe that the number of transitions of temporal paths or walks is constant in the reductions of \cref{thm:n-star} and \cref{thm:w1hardnessVCN}, indicating that finding solutions with few transitions is still hard. We believe that the duration might be a more promising parameter, since it is large in the reductions of \cref{thm:w1hardnessVCN} and \cref{thm:k-star}, leading to the following question.
\begin{itemize}
    \item Is \AllProbs\ in FPT or W[1]-hard when parameterized by the combination of the number of source-sink pairs and the maximum duration of any temporal (path/walk) in the solution?
\end{itemize}

Finally, we leave open whether our results also hold for the non-strict case, where temporal (paths/walks) use transitions with non-decreasing (instead of increasing) time labels. We conjecture that all our results can be adapted for this case.

\bibliographystyle{abbrvnat}
\bibliography{bibliography}	

\end{document}